%% file: main.tex
\documentclass[runningheads]{llncs}

\usepackage{etoolbox}
\newtoggle{arxiv}
\toggletrue{arxiv} %

\usepackage{graphicx}
\usepackage{amsmath}
\usepackage{amssymb}
\usepackage{color}
\usepackage{xcolor}
\usepackage{booktabs}
\usepackage{subcaption}
\usepackage{adjustbox}
\usepackage{cancel}
\usepackage{csquotes}
\usepackage{xfrac}
\usepackage{xspace}
\usepackage{mathtools}
\usepackage{stmaryrd}
\usepackage{braket}
\iftoggle{arxiv}{\setcounter{tocdepth}{3}}{}
\usepackage[
	pdfencoding=auto,
	psdextra,
	colorlinks=true,
	linkcolor=blue!50!black,
	citecolor=blue!50!black,
	urlcolor=blue!50!black,
	bookmarks=true,
	bookmarksopen=true,
	bookmarksnumbered=true,
	pdfstartview=FitH
]{hyperref}
\usepackage[nameinlink,noabbrev]{cleveref}
\usepackage{nicefrac}
\usepackage{proof}
\usepackage{pgfplots}
\pgfplotsset{compat=1.18}
\usepackage{thm-restate}
\usepackage{listings}
\usepackage{float}
\usepackage{multirow}
\usepackage{tikz}
\usetikzlibrary{arrows,automata,positioning,shapes}
\usepackage{wrapfig}
\usepackage{quiver}
\iftoggle{arxiv}{\usepackage{orcidlink}}{}

\usepackage[
  backend=biber,
  style=lncs
]{biblatex}

\newcommand{\authororcid}[1]{\iftoggle{arxiv}{\orcidlink{#1}}{\orcidID{#1}}}

\input{macros}

\newif\ifmanualqed
\makeatletter
\let\llncsproof\proof
\let\endllncsproof\endproof
\renewcommand{\squareforqed}{\ensuremath{\square}}
\renewcommand{\qedhere}{\global\manualqedtrue\hfill\ensuremath{\squareforqed}}
\renewenvironment{proof}[1][]
{\global\manualqedfalse%
 \if\relax\detokenize{#1}\relax
   \llncsproof
 \else
   \llncsproof[#1]%
 \fi}
{\ifmanualqed\else\qed\fi\endllncsproof}
\makeatother

\lstdefinelanguage{HeyVL}{
	morekeywords={var,while,if,else,@invariant,@slice_verify,@slice_error,Int,Bool,EUReal,UInt,proc,coproc,pre,post,flip,coassume,assume,assert,coassert,\cup,cohavoc,havoc,domain,func,UReal,reward,ite,forall,axiom,true,false},
	morecomment=[l]{//},
	morecomment=[s]{/*}{*/},
	sensitive=true
}

\AfterEndEnvironment{definition}{\par\noindent\ignorespaces}
\AfterEndEnvironment{lemma}{\par\noindent\ignorespaces}
\AfterEndEnvironment{theorem}{\par\noindent\ignorespaces}
\AfterEndEnvironment{corollary}{\par\noindent\ignorespaces}

\begin{document}

\title{Highly Incremental: A Simple Programmatic Approach for Many Objectives\iftoggle{arxiv}{\texorpdfstring{\\}{ }(Extended Version)}{}}
\titlerunning{A Simple Programmatic Approach for Many Objectives}

 \author{Philipp Schröer\inst{1}\authororcid{0000-0002-4329-530X} \and Joost-Pieter Katoen\inst{1}\authororcid{0000-0002-6143-1926}}
 \authorrunning{P. Schröer and J.-P. Katoen}
 \institute{RWTH Aachen University, Germany \\
 	\email{\nolinkurl{phisch@cs.rwth-aachen.de}, \nolinkurl{katoen@cs.rwth-aachen.de}}}

\maketitle

\input{sections/0_abstract.tex}

\input{sections/1_introduction.tex}
\input{sections/2_background.tex}
\input{sections/3_reward_structures.tex}
\input{sections/4_encoding_moments_in_rewards.tex}
\input{sections/5_transformation_functions.tex}
\input{sections/6_automation.tex}
\input{sections/7_related_work.tex}

\input{sections/8_conclusion.tex}

\begin{credits}
\subsubsection*{\ackname}
This work was partially supported by the ERC POC Grant VERIPROB (grant no.\ 101158076) and by the European Union's Horizon 2020 research and innovation programme under the Marie Sk\l{}odowska-Curie grant agreement MISSION (grant no.\ 101008233).

\end{credits}

\printbibliography

\iftoggle{arxiv}{
\newpage
\appendix
\input{appendices/omitted_proofs.tex}
\input{appendices/case_studies.tex}
\clearpage
\input{appendices/more_related_work.tex}

}{}

\end{document}

%% file: macros.tex
\newcommand{\morespace}[1]{~{}#1{}~}

\newcommand{\eeq}{\morespace{=}}

\usepackage{xcolor}
\definecolor{hintgray}{RGB}{92,92,92}

\definecolor{col1}{RGB}{254,195,10}
\definecolor{col2}{RGB}{251,0,6}
\definecolor{col3}{RGB}{252,11,135}
\definecolor{col4}{RGB}{17,139,2}
\definecolor{col5}{RGB}{14,131,254}
\definecolor{col6}{RGB}{82,0,135}

\definecolor{brLightGreen}{HTML}{addd8e}
\definecolor{slidesBackground}{RGB}{250,250,250}
\definecolor{brLightGray}{HTML}{f0f0f0}

\definecolor{orange}{RGB}{230,159,0}
\definecolor{skyblue}{RGB}{86,180,233}
\definecolor{brBlue}{RGB}{0,114,178}
\definecolor{bluishgreen}{RGB}{0,158,115}
\definecolor{vermillion}{RGB}{213,94,0}
\definecolor{reddishpurple}{RGB}{204,121,167}

\colorlet{heyvlColor}{vermillion!60!black}
\definecolor{prepostColor}{RGB}{0,69,107}
\colorlet{stmtColor}{bluishgreen!70!black}

\newcommand{\extVersionCite}{\cite{Schroer:1029645}}
\newcommand{\appRefOrExt}[1]{\iftoggle{arxiv}{\Cref{#1}}{\extVersionCite}}

\newcommand{\colheylo}[1]{{\color{prepostColor}#1}}

\newcommand{\expa}{\colheylo{\ensuremath{X}}}
\newcommand{\expb}{\colheylo{\ensuremath{Y}}}

\newcommand{\expleq}{\ensuremath{\preceq}}

\newcommand{\substBy}[2]{[{#1} \mapsto {#2}]}

\newcommand{\iverson}[1]{\left[{#1}\right]}

\newcommand{\lam}[2]{\lambda {#1}.~{#2}}

\newcommand{\monus}{~\dot{-}~}

\newcommand{\lfp}[2]{\mathrm{lfp}\,{#1}.~{#2}}

\newcommand{\sfsymbol}[1]{\textsf{\upshape {#1}}}
\newcommand{\Vars}{\sfsymbol{Vars}\xspace}

\newcommand{\pGCL}{\sfsymbol{pGCL}\xspace}

\newcommand{\ArithExp}{\sfsymbol{AExp}\xspace}
\newcommand{\BExp}{\sfsymbol{BExp}\xspace}

\newcommand{\aexpr}{\ensuremath{a}} %
\newcommand{\bexpr}{\ensuremath{b}} %

\newcommand{\Nats}{\mathbb{N}}

\newcommand{\Reals}{\mathbb{R}}
\newcommand{\PosReals}{\mathbb{R}_{\geq 0}}
\newcommand{\PosRealsInf}{\mathbb{R}_{\geq 0}^{\infty}}

\newcommand{\States}{\mathsf{States}}
\newcommand{\State}{\sigma}
\newcommand{\Vals}{\mathsf{Vals}}

\newcommand{\symStmt}[1]{\ensuremath{{\color{stmtColor}{#1}}}}
\newcommand{\symSkip}{\symStmt{\texttt{skip}}}
\newcommand{\symAssign}{\,{\coloneqq}\,}

\newcommand{\symSemi}{\symStmt{\texttt{;}~}}
\newcommand{\symIf}{\symStmt{\texttt{if}}}
\newcommand{\symElse}{\symStmt{\texttt{else}}}
\newcommand{\symAssert}{\symStmt{\texttt{assert}}}

\newcommand{\symWhile}{\symStmt{\texttt{while}}}

\newcommand{\symTick}{\symStmt{\texttt{reward}}}
\newcommand{\symProc}{\symStmt{\texttt{proc}}}
\newcommand{\symcoProc}{\symStmt{\texttt{coproc}}}

\newcommand{\symPre}{\symStmt{\texttt{pre}}}

\newcommand{\blockStart}{\ensuremath{\{}}
\newcommand{\blockEnd}{\ensuremath{\}}}

\newcommand{\Requires}[1]{\symPre~{\color{prepostColor}#1}}

\newcommand{\multi}[1]{\ensuremath{\overline{#1}}}

\newcommand{\stmt}{{\color{heyvlColor}\ensuremath{S}}}
\newcommand{\stmtP}{{\color{heyvlColor}\ensuremath{S'}}}

\newcommand{\stmtOne}{{\color{heyvlColor}\ensuremath{S_1}}}

\newcommand{\stmtTwo}{{\color{heyvlColor}\ensuremath{S_2}}}

\newcommand{\stmtSkip}{\symSkip}
\newcommand{\stmtAsgn}[2]{\ensuremath{{#1} \symAssign {#2}}}

\newcommand{\stmtSeq}[2]{\ensuremath{{#1}\symSemi {#2}}}
\newcommand{\stmtIf}[3]{\ensuremath{\symIf~({#1})~\{ {#2} \}~\symElse~\{ {#3} \}}}
\newcommand{\stmtIfStart}[1]{\ensuremath{\symIf~({#1})~\{ }}
\newcommand{\stmtProb}[3]{\ensuremath{\{ {#2} \} ~[{#1}]~ \{ {#3} \}}}
\newcommand{\stmtWhile}[2]{\ensuremath{\symWhile~({#1})~\{~ {#2} ~\}}}
\newcommand{\headerWhile}[1]{\ensuremath{\symWhile~({#1})}}

\newcommand{\stmtTick}[1]{\ensuremath{\symTick~{#1}}}
\newcommand{\stmtReward}[1]{\stmtTick{#1}}

\newcommand{\exprTrue}{\ensuremath{\texttt{true}}}
\newcommand{\exprFalse}{\ensuremath{\texttt{false}}}
\newcommand{\exprFlip}[1]{\ensuremath{\symStmt{\texttt{flip}}({#1})}}

\newcommand{\eval}[1]{\llbracket{#1}\rrbracket}

\newcommand{\symWp}{\sfsymbol{wp}}

\renewcommand{\wp}[1]{\symWp\llbracket{#1}\rrbracket}
\newcommand{\wpleftright}[1]{\symWp\left\llbracket{#1}\right\rrbracket}

\newcommand{\symErt}{\sfsymbol{ert}}
\newcommand{\ert}[1]{\symErt\llbracket{#1}\rrbracket}

\newcommand{\symAert}{\sfsymbol{aert}}

\newcommand{\Expectations}{\mathbb{E}}

\newcommand{\rt}{\tau}
\newcommand{\rwTrans}[2]{\mathcal{T}_{#1}(#2)}
\newcommand{\rwTransP}[2]{\mathcal{T}'_{#1}(#2)}

\newcommand{\Expected}{\mathsf{E}}

\newcommand{\Conf}{\mathfrak{S}}
\newcommand{\conf}{s}

\newcommand{\Dist}[1]{\mathsf{Dist}({#1})}

\newcommand{\PathsN}[2]{\mathsf{Paths}^{#2}_{#1}}
\newcommand{\progMdp}[2]{\mathcal{M}_{#1}^{#2}}

\newcommand{\Prob}{\mathsf{Prob}}
\newcommand{\ProbP}{\mathsf{Prob'}}

\newcommand{\symRew}{\mathsf{rew}}
\newcommand{\Rew}[1]{\symRew({#1})}
\newcommand{\trans}[1]{\xrightarrow{#1}}

\newcommand{\term}{{\Downarrow}}
\newcommand{\final}{\bot}

\newcommand{\symRewP}{\mathsf{rew}'}
\newcommand{\RewP}[1]{\symRewP({#1})}

\newcommand{\RewState}[2]{\symRew^{#1}({#2})}
\newcommand{\RewInd}[2]{\symRew_{#1}({#2})}
\newcommand{\RewStateInd}[3]{\symRew^{#1}_{#2}({#3})}

\newcommand{\qedhere}{\hfill\ensuremath{\squareforqed}}

\definecolor{mydarkgreen}{RGB}{64,150,64}
\newcommand{\highlight}[1]{\colorbox{white!94!mydarkgreen}{$\displaystyle#1$\hspace{-0.5em}}}

\renewcommand{\monus}{\ensuremath{\ominus}}

%% file: sections/0_abstract.tex
\begin{abstract}
We present a one-fits-all programmatic approach to reason about a
plethora of objectives on probabilistic programs. The first ingredient
is to add a reward-statement to the language. We then define a program
transformation applying a monotone function to the cumulative reward of
the program. The key idea is that this transformation uses incremental
differences in the reward. This simple, elegant approach enables to
express e.g., higher moments, threshold probabilities
of rewards, the expected excess over a budget, and
moment-generating functions. All these objectives can now be
analyzed using a single existing approach: probabilistic wp-reasoning.
We automated verification using the Caesar deductive verifier
and report on the application of the transformation to some examples.
\end{abstract}

%% file: sections/1_introduction.tex
\section{Introduction}\label{sec:introduction}

\paragraph{Probabilistic programs.}
Probabilistic programs combine ordinary programming structures such as loops and conditionals with probabilistic choices.
They are a powerful modeling tool for a wide range of applications, including randomized algorithms, security mechanisms, and probabilistic databases.
Due to the intrinsic randomness, the runtime of a probabilistic program follows a probability distribution.
\Cref{fig:example-program-a} models a web server retrying a database call until success. Variable \texttt{done} records success; each loop iteration yields a \stmtReward{1} (1 second). Each iteration succeeds with probability $\nicefrac{1}{2}$ and the program terminates with probability one (\emph{almost-sure termination}).
Its expected runtime is $2$ (\emph{positive almost-sure termination}).

\paragraph{Same expected runtime, different variance.}
We add a cache (\Cref{fig:example-program-b}) raising the success probability to $\nicefrac{2}{3}$ at an initialization cost of $\nicefrac{1}{2}$ seconds.
The new program also terminates almost surely, and perhaps more surprisingly, its expected runtime is also $2$.
Despite identical expected runtimes, their runtime distributions differ (\Cref{fig:example-programs-distribution}).
The probability that \Cref{fig:example-program-a} terminates earlier is higher, while \Cref{fig:example-program-b} concentrates more tightly around 2s.
The variances are $2\,\text{s}^2$ (no cache) versus $\nicefrac{3}{4}\,\text{s}^2$ (cache), making the runtime of the cached variant more predictable.
The example shows that the expected value -- an \emph{average} -- is not always sufficient to capture the behavior of a probabilistic program.

\begin{figure}[t]
    \centering
    \begin{subfigure}[t]{0.45\textwidth}
        \centering
        \begin{align*}
             &                                     \\
             & \stmtAsgn{done}{\exprFalse}\symSemi \\
             & \symWhile~(\neg done)~\blockStart   \\
             & \quad \stmtReward{1}\symSemi        \\ %
             & \quad
            \underset{\text{\emph{probabilistic choice}}}{
                \stmtProb{\nicefrac{1}{2}}{\stmtAsgn{done}{\exprTrue}}{\stmtSkip}
            }
            \\[-1em]
             & \blockEnd
        \end{align*}
        \caption{Program modeling a web server which retries a database call, succeeding with probability $\nicefrac{1}{2}$. Each database call takes 1 second.}
        \label{fig:example-program-a}
    \end{subfigure}
    \hfill
    \begin{subfigure}[t]{0.45\textwidth}
        \centering
        \begin{align*}
             & \stmtReward{\nicefrac{1}{2}}\symSemi \\ %
             & \stmtAsgn{done}{\exprFalse}\symSemi  \\
             & \symWhile~(\neg done)~\blockStart    \\
             & \quad \stmtReward{1}\symSemi         \\ %
             & \quad
            \underset{\text{\emph{probabilistic choice}}}{
                \stmtProb{\nicefrac{2}{3}}{\stmtAsgn{done}{\exprTrue}}{\stmtSkip}
            }
            \\[-1em]
             & \blockEnd
        \end{align*}
        \caption{Variant with a cache. It initializes a cache (taking $\nicefrac{1}{2}$ seconds) and then retries a database call which succeeds with probability $\nicefrac{2}{3}$. Each database call takes $1$ second.}
        \label{fig:example-program-b}
    \end{subfigure}
    \caption{Example programs modeling a web server retrying a database call.}
    \label{fig:example-programs}
\end{figure}

\begin{figure}[t]
    \centering
    {
        \tiny{}
        \begin{tikzpicture}
            \begin{axis}[
                    ybar,
                    ymin=0,
                    xlabel={Runtime},
                    ylabel={Probability},
                    xtick={1,1.5,...,7.5},
                    bar width=0.4cm,
                    bar shift=0pt,
                    width=0.8\textwidth,
                    height=2.5cm,
                    enlarge x limits=0.15,
                ]
                \addplot [
                    red,
                    fill=red!90!black,
                    opacity=0.5,
                    domain=1:7,
                    samples=7,
                ] {(1-0.5)^x};
                \addlegendentry{\Cref{fig:example-program-a}};
                \addplot [
                    blue,
                    fill=blue!90!black,
                    opacity=0.5,
                    domain=1.5:7.5,
                    samples=7,
                ] {(2/3) * (1/3)^(x-1.5)};
                \addlegendentry{\Cref{fig:example-program-b}};
            \end{axis}
        \end{tikzpicture}
    }
    \caption{Probability distribution of runtimes of the programs in \Cref{fig:example-programs}.}
    \label{fig:example-programs-distribution}
\end{figure}
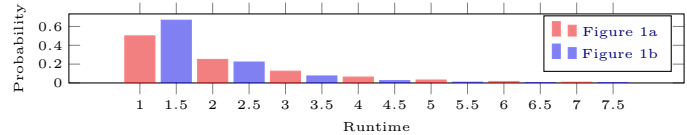

\paragraph{Deductive verification of probabilistic programs.}
Dedicated techniques have been developed to reason about different objectives of probabilistic programs.
The weakest pre-expectation ($\symWp$) semantics \cite{DBLP:series/mcs/McIverM05} is a well-established formalism to reason about expected values on termination.
An extensive catalog of proof rules has been developed, including \cite{DBLP:conf/esop/BatzKKM18,DBLP:journals/pacmpl/McIverMKK18,DBLP:journals/pacmpl/HarkKGK20,DBLP:conf/cav/BatzCKKMS20,DBLP:journals/pacmpl/FengCSKKZ23}.
The $\symWp$ calculus has been extended for reasoning about expected runtimes~\cite{DBLP:conf/esop/KaminskiKMO16,DBLP:books/cu/20/KaminskiKM20} and general rewards~\cite{DBLP:conf/birthday/BatzKM025}.
Tool-assisted $\symWp$-based verification includes \emph{HOL} mechanization~\cite{DBLP:journals/tcs/HurdMM05} and the \emph{Caesar} deductive verifier~\cite{DBLP:journals/pacmpl/SchroerBKKM23}.
Tailored variations have been proposed for covariances~\cite{DBLP:conf/qest/KaminskiKM16}, higher moments of rewards~\cite{DBLP:conf/tacas/KuraUH19,DBLP:conf/pldi/Wang0R21}, and for amortized complexity analysis~\cite{DBLP:journals/pacmpl/BatzKKMV23}.

\paragraph{Programmatic rewards as a unifying abstraction.}
Different objectives currently require separate $\symWp$-calculi, hindering reuse.
We aim for a \emph{one-fits-all} approach so that existing $\symWp$ semantics, proof rules, and tools can be reused for a broad range of objectives.
We do so by adding a $\symTick$ statement to the programming language, inspired by~\cite{DBLP:journals/pacmpl/AvanziniMS20,DBLP:journals/pacmpl/BatzKKMV23}, to collect (non-negative) \emph{rewards} during program execution.
The \emph{expected cumulative reward} is then the expected sum of all rewards collected during a program run.
We show how different objectives can be modeled using reward statements, including expected runtimes, expected visiting times, and discounted rewards.

\paragraph{Incremental differences for reward transformations.}
To tackle objectives like higher moments of expected rewards, we need to reason about \emph{transformed} expected rewards of program $\stmt$ such as $\Expected(\Rew{\stmt}^2)$.
Our key idea is to make use of \emph{incremental differences} of rewards, inspired by the $\symAert$ calculus~\cite{DBLP:journals/pacmpl/BatzKKMV23}.
We define a simple program transformation that reduces reasoning about complex expected rewards to reasoning about standard expected rewards, allowing us to reuse existing $\symWp$ theory, techniques, and tools.

\paragraph{Contributions and outline.}
We provide one programming language and one weakest pre-expectation semantics to reason about a broad range of quantitative objectives of probabilistic programs in a unified way.
\Cref{sec:foundations} introduces the syntax and semantics for probabilistic programs.
Our main contributions are:
\begin{enumerate}
    \item A programmatic encoding for multiple quantitative objectives (\Cref{sec:programmatic-reward-modeling}),
    \item A program transformation (\Cref{sec:program-transformation}) that is proven sound w.r.t. a Markov chain transformation (\Cref{sec:mc-transformation}),
    \item Application to different objectives, including higher moments. (\Cref{sec:reward-transformation-functions}),
\end{enumerate}
Finally, we briefly touch on automation in \Cref{sec:automation}, compare related work in \Cref{sec:related-work} and conclude in \Cref{sec:conclusion}.

%% file: sections/2_background.tex
\section{Foundations: Semantics for Probabilistic Programs}\label{sec:foundations}

In this section, we introduce the syntax and semantics of probabilistic programs with $\symTick$ statements.
Our presentation is based on \cite{DBLP:conf/birthday/BatzKM025}, with two key adaptations: (1) we allow arbitrary expressions $\aexpr$ in $\stmtTick{\aexpr}$ statements instead of only constants, and (2) we allow program variables to take on values from the set $\PosRealsInf = \mathbb{R}_{\geq 0} \cup \Set{\infty}$.
We formally define the syntax of probabilistic guarded command language ($\pGCL$) programs, describe their operational semantics via Markov chains, and present the weakest pre-expectation semantics.

\subsection{Syntax of Probabilistic Programs}

\paragraph{States and expressions.}
We consider a finite set of variables $\Vars$.
A \emph{program state} $\State$ is a mapping from variables to values, i.e., $\State \in \States = \Vars \to \Vals$.
In this work, we will use $\Vals = \PosRealsInf$.
We write $\State\substBy{x}{e}$ for the state $\State$ updated so that $x$ maps to the value of $e$ in $\State$ and all other variables are untouched, i.e. $\State\substBy{x}{e}(y) = e(\State)$ if $y = x$, and $\State(y)$ otherwise.
An arithmetic expression $\aexpr \in \ArithExp$ is a mapping from program states to values, i.e. $\ArithExp = \States \to \Vals$.
Boolean expressions $\bexpr \in \BExp$ map program states to $\Set{\exprTrue,~ \exprFalse}$.
\paragraph{Programs.}
The syntax of $\pGCL$ programs $\stmt$ is defined by the following grammar, where $x \in \Vars$, $\aexpr \in \ArithExp$, $p \in [0,1]$, and $\bexpr \in \BExp$:
\[
    \begin{array}{rcl}
        \stmt & ::=  & \stmtSkip \mid \stmtAsgn{x}{\aexpr} \mid \stmtTick{\aexpr} \mid \stmtOne \symSemi \stmtTwo                 \\
              & \mid & \stmtProb{p}{\stmtOne}{\stmtTwo} \mid \stmtIf{\bexpr}{\stmtOne}{\stmtTwo} \mid \stmtWhile{\bexpr}{\stmt}~.
    \end{array}
\]
The $\stmtSkip$ statement does nothing.
The assignment $\stmtAsgn{x}{\aexpr}$ assigns the value of expression $\aexpr$ to variable $x$.
The $\stmtTick{\aexpr}$ statement collects a non-negative reward $\aexpr$ when executed.
Note that $\aexpr$ is an expression that may contain program variables.
The sequential composition $\stmtOne \symSemi \stmtTwo$ executes $\stmtOne$ followed by $\stmtTwo$.
The probabilistic choice $\stmtProb{p}{\stmtOne}{\stmtTwo}$ executes $\stmtOne$ with probability $p$ and $\stmtTwo$ with probability $1-p$.
The conditional $\stmtIf{\bexpr}{\stmtOne}{\stmtTwo}$ executes $\stmtOne$ if the boolean expression $\bexpr$ evaluates to true, and $\stmtTwo$ otherwise.
The loop $\stmtWhile{\bexpr}{\stmt}$ repeatedly executes $\stmt$ as long as the boolean expression $\bexpr$ evaluates to true.

\subsection{Operational Semantics}

\begin{wrapfigure}{r}{0.25\textwidth}
    \centering
    \vspace*{-0.7cm}
    {\footnotesize
    \begin{tikzpicture}[
        state/.style={circle, draw},
    ]
        \node[state, initial, initial where=left, initial text={}] (s1) [] {$s_1$};
        \node[state] (s2) [below of=s1] {$s_2$};
        \node[state] (s3) [left of=s2] {$s_3$};
        \node[above=-0.0cm of s3, font=\scriptsize] {$\Rew{{s_3}} {=} 1$};
        \node[state] (s4) [below left=0.5cm and 0.0cm of s3] {$s_5$};
        \node[state] (s5) [below right of=s4] {$s_6$};
        \node[state](s6) [below right=0.5cm and 0.0cm of s3] {$s_4$};
        \node[state] (term) [below right of=s6] {$s_7$};
        \node[state] (final) [below left of=term] {$\bot$};

        \path[->] (s1) edge node {} (s2);
        \path[->] (s2) edge node {} (s3);
        \path[->] (s3) edge node[left] {\scriptsize{}0.5} (s4);
        \path[->] (s3) edge node[right] {\scriptsize{}0.5} (s6);
        \path[->] (s4) edge node {} (s5);
        \path[->] (s5) edge node {} (term);
        \path[->] (term) edge node {} (final);
        \path[->] (final) edge[loop left] node {} (final);
        \path[->,bend right=60] (s6) edge node {} (s2);
    \end{tikzpicture}}
    \caption{Operational MC for \Cref{fig:example-program-a}.}
    \label{fig:webserver-markov-chain}
    \vspace*{-0.5cm}
\end{wrapfigure}

The operational semantics of programs is given by Markov chains with rewards, whose foundations we briefly recall.

A \emph{Markov chain} is a tuple $\mathcal{M} = (\Conf, \Prob, \State_1, \symRew)$, where $\Conf$ is a (countable) set of configurations, $\Prob \colon \Conf \to \Dist{\Conf}$ is the transition function, $\State_1 \in \Conf$ is the initial configuration, and $\symRew \colon \Conf \to \PosRealsInf$ is the reward function.

A (finite) \emph{path} $\pi \in \PathsN{\mathcal{M}}{=n}$ of length $n$ in $\mathcal{M}$ is a sequence of configurations $\pi = \conf_1 \conf_2 \ldots \conf_n$ with $\Prob(\conf_i)(\conf_{i+1}) > 0$ for all $1 \leq i < n$.
We lift $\Prob$ and $\symRew$ to finite paths by $\Prob(\conf_1 \ldots \conf_n) = \prod_{1 \leq i < n} \Prob(\conf_i)(\conf_{i+1})$ and $\Rew{\conf_1 \ldots \conf_n} = \sum_{1 \leq i \leq n} \Rew{\conf_i}$.
Let $\Expected(\Rew{\mathcal{M}})$ denote the expected cumulative reward of $\mathcal{M}$, given by $\sup_{n \in \Nats} \sum_{\pi \in \PathsN{\mathcal{M}}{=n}(\State)} \Prob(\pi) \cdot \Rew{\pi}$.
More generally, the reward transformed by a function $f \colon \PosRealsInf \to \PosRealsInf$ is defined as $\Expected(f(\Rew{\mathcal{M}})) = \sup_{n \in \Nats} \sum_{\pi \in \PathsN{\mathcal{M}}{=n}(\State)} \Prob(\pi) \cdot f(\Rew{\pi})$.

The \emph{operational Markov chain} $\progMdp{\stmt}{\State_1} = (\Conf, \Prob, \conf_1, \symRew)$ of program $\stmt$ is given by the reachable\footnote{The inference rules admit only finitely many transitions from each configuration, hence the set of reachable configurations is countable.} configurations $\Conf \subseteq (\pGCL \cup \set{\term}) \times \States \cup \Set{\final}$ with transition probabilities $\Prob$ (see \appRefOrExt{fig:operational-semantics}), initial configuration $\conf_1 = (\stmt, \State_1)$, and reward function $\symRew$ defined as follows:
\begin{align*}
    \Rew{(\stmt, \State)} = \begin{cases}
        r & \text{if } (\exists \stmt' \in \pGCL.~ \stmt = \stmtTick{r}\symSemi \stmt') \lor (\stmt = \stmtTick{r}) \\
        0 & \text{otherwise}
    \end{cases}
\end{align*}
The operational Markov chain for the program in \Cref{fig:example-program-a} is shown in \Cref{fig:webserver-markov-chain}.
We write $\Expected(\RewState{\State_1}{\stmt})$ to denote the expected cumulative reward of the operational Markov chain $\progMdp{\stmt}{\State_1}$, i.e., $\Expected(\Rew{\progMdp{\stmt}{\State_1}})$.
Let $\Expected(\Rew{\stmt}) = \lam{\State}{\Expected(\RewState{\State}{\stmt})}$.

\subsection{Weakest Pre-Expectation Semantics}\label{sec:wp-semantics}

An \emph{expectation} is a function $f \colon \States \to \PosRealsInf$.
Let $\Expectations$ denote the set of all expectations.
Expectations form a complete lattice $(\Expectations, \expleq)$ with pointwise order $\expa \expleq \expb$ iff $\expa(\State) \leq \expb(\State)$ for all $\State \in \States$.
Thus, all infima and suprema exist, and by Kleene's fixed-point theorem, least fixed points $\lfp{\expb}{\Phi(\expb)}$ of monotone functions $\Phi \colon \Expectations \to \Expectations$ exist.
Let $\iverson{\bexpr} \in \Expectations$ map a Boolean expression $\bexpr$ to $\iverson{\bexpr}(\State) = 1$ iff $\bexpr(\State) = \exprTrue$, and to $0$ otherwise.
A substitution $\expa\substBy{x}{\aexpr}$ of variable $x$ by expression $\aexpr$ in $\expa \in \Expectations$ is given by $\expa\substBy{x}{\aexpr}(\State) = \expa(\State\substBy{x}{\aexpr})$ for all $\State \in \States$.
Arithmetic on expectations is lifted pointwise, e.g. $(\expa+\expb)(\State)=\expa(\State)+\expb(\State)$.

The \emph{weakest pre-expectation} for a program $\stmt$ is a function $\wp{\stmt} \colon \Expectations \to \Expectations$ defined by structural induction on $\stmt$; see \Cref{fig:wp-semantics}.
Intuitively, $\wp{\stmt}(\expa)$ is the expected cumulative reward of the program $\stmt$ starting from the initial state $\State$, given that the reward earned on termination is $\expa$.
The semantics is standard\footnote{The $\symWp$ calculus was originally defined without $\symTick$ statements~\cite{DBLP:series/mcs/McIverM05,DBLP:phd/dnb/Kaminski19};
we use the extended calculus of~\cite{DBLP:conf/birthday/BatzKM025} with $\symTick$ statements.}, but we highlight $\symTick$ and while loops.
For a reward statement $\stmtTick{\aexpr}$, we have $\wp{\stmtTick{\aexpr}}(\expa) = \aexpr + \expa$, i.e. the result of evaluating $\aexpr$ in the current state is added to $\expa$.
For a loop $\stmtWhile{\bexpr}{\stmt}$, the pre-expectation is given by the least fixed point of the functional $\expb \mapsto \iverson{\bexpr} \cdot \wp{\stmt}(\expb) + \iverson{\neg \bexpr} \cdot \expa$, dubbed the \emph{loop-characteristic functional w.r.t. post $\expa$}.

\begin{figure}[t]
    \renewcommand{\arraystretch}{1.2}%
    \setlength{\tabcolsep}{6pt}%
    \begin{tabular}{@{}l@{\hspace{1.0em}}l | l@{\hspace{1.0em}}l@{}}
        \toprule
        $\stmt$                      & $\wp{\stmt}(\expa)$                   & $\stmt$                                        & $\wp{\stmt}(\expa)$                                                                                                        \\
        \midrule
        $\stmtSkip$                  & $\expa$                               & $\stmtProb{p}{\stmtOne}{\stmtTwo}$             & $p \,{\cdot}\, \wp{\stmtOne}(\expa) {+} (1{-}p) \,{\cdot}\, \wp{\stmtTwo}(\expa)$                                          \\
        $\stmtAsgn{x}{\aexpr}$       & $\expa\substBy{x}{\aexpr}$            & $\stmtIfStart{\bexpr}~\stmtOne~\blockEnd$      & \multirow{2}{*}{$\iverson{\bexpr} {\cdot}\, \wp{\stmtOne}(\expa) {+} \iverson{\neg\bexpr} {\cdot}\, \wp{\stmtTwo}(\expa)$} \\
        $\stmtTick{\aexpr}$          & $\aexpr + \expa$                      & $\quad\symElse~\blockStart~\stmtTwo~\blockEnd$ &                                                                                                                            \\
        $\stmtOne \symSemi \stmtTwo$ & $\wp{\stmtOne}(\wp{\stmtTwo}(\expa))$ & $\stmtWhile{\bexpr}{\stmtP}$                    & $\lfp{\expb}{\iverson{\bexpr} \cdot \wp{\stmtP}(\expb) + \iverson{\neg \bexpr} \cdot \expa}$                                \\
        \bottomrule
    \end{tabular}
    \caption{Inductive definition of the \emph{weakest pre-expectation semantics} for program $\stmt$, $\wp{\stmt} \colon \Expectations \to \Expectations$ where $\expa \in \Expectations$ is the \emph{post-expectation}.}
    \label{fig:wp-semantics}
\end{figure}

For example, the least fixed point for the loop in \Cref{fig:example-program-a} with post $\expa = 0$ is given by $(\lfp{\expb}{\iverson{\neg done} \cdot (1 + 0.5 \cdot \expb\substBy{done}{\exprTrue} + 0.5 \cdot \expb)}) = \iverson{\neg done} \cdot 2$.
For the full program $\stmt$, we have $\wp{\stmt}(0) = (\iverson{\neg done} \cdot 2)\substBy{done}{\exprFalse} = 2$.

\paragraph{Soundness.}
The expression $\wp{\stmt}(\expa)(\State)$ evaluates to the expected cumulative reward in the Markov chain corresponding to the semantics of $\stmt$ starting in $\State$~\cite[Theorem 6]{DBLP:conf/birthday/BatzKM025}, i.e. $\wp{\stmt}(\expa)(\State) = \Expected(\RewStateInd{\State}{\expa}{\stmt})$.
Here, $\RewStateInd{\State}{\expa}{\stmt}$ is the \emph{reward function induced by $\expa$} on top of the operational Markov chain $\progMdp{\stmt}{\State}$, which collects the reward $\expa(\State)$ upon termination, and otherwise the reward defined for each statement $\stmt$ in configuration $(\stmt, \State)$ according to \appRefOrExt{fig:operational-semantics}.
Formally, it is defined as $\RewInd{\expa}{(\term,\State)} = \expa(\State)$ and $\RewInd{\expa}{(\stmt,\State)} = \Rew{(\stmt,\State)}$ otherwise.
We use abbreviations $\RewStateInd{\State}{\expa}{\stmt} = \RewInd{\expa}{\progMdp{\stmt}{\State}}$ and $\Expected(\RewInd{\expa}{\stmt}) = \lam{\State}{\Expected(\RewStateInd{\State}{\expa}{\stmt})}$.

\paragraph{Healthiness conditions.}
Corresponding to properties of expected values, $\symWp$ satisfies several well-behavedness properties, often called \emph{healthiness conditions}~\cite{DBLP:series/mcs/McIverM05,DBLP:phd/dnb/Kaminski19,DBLP:conf/birthday/BatzKM025}.
Let $\stmt \in \pGCL$ be a program and $\expa, \expb \in \Expectations$ be expectations.
Then, $\wp{\stmt}$ is monotonic, i.e. for $\expa \expleq \expb$, we have $\wp{\stmt}(\expa) \expleq \wp{\stmt}(\expb)$.
Further, $\wp{\stmt}$ is linear for $\symTick$-free programs $\stmt$, i.e. for all $\expa, \expb \in \Expectations$ and $a \in \PosReals$, we have $\wp{\stmt}(a \cdot \expa + \expb) = a \cdot \wp{\stmt}(\expa) + \wp{\stmt}(\expb)$.

%% file: sections/3_reward_structures.tex
\section{Programmatic Reward Modeling}\label{sec:programmatic-reward-modeling}

We introduce \emph{programmatic reward modeling} for probabilistic programs via the program statement $\symTick$.
By expressing rewards via $\symTick$, we can directly incorporate different reward structures into the program text.

\subsection{Rewards on Termination}\label{sec:rewards-on-termination}

A common reward structure collects rewards only on termination, as considered by the original weakest pre-expectation semantics $\symWp$~\cite{DBLP:series/mcs/McIverM05,DBLP:journals/pe/GretzKM14,DBLP:phd/dnb/Kaminski19} (without $\symTick$ statements).
Here, for a program $\stmt$ and $f \in \Expectations$, $\wp{\stmt}(f)$ is the expected value of $f$ on termination of $\stmt$.
Using $\symWp$, we can also easily query for transformed expected rewards on termination, e.g. $\wp{\stmt}(x^2)$ to obtain the second moment of the value of $x$ on termination of $\stmt$.

For a more programmatic approach, we generalize the classical idea of embedding a post-condition into a program via an $\symAssert$ statement, which eliminates the need to specify the post-condition separately.
Instead of specifying $\stmt$ explicitly with some post-expectation $\expa$, we can use the $\symTick$ statement to collect rewards on termination, i.e. the sequence $\stmt\symSemi \stmtTick{\expa}$, and use $\expa = 0$.
The corresponding soundness statement is also simplified, replacing $\symRew_\expa$ by $\symRew$.
Here, it is essential that we allow arbitrary arithmetic functions $\expa$ in the $\symTick$ statement, including expectations $\Expectations$ (c.f. \Cref{sec:foundations}).\footnote{Although not necessary for our setup, we conjecture that the assertion language from~\cite{DBLP:journals/pacmpl/BatzKKM21} can be extended to support $\symTick$ statements while remaining \emph{expressive}, i.e. for all syntactic $\expa \in \Expectations$, one can syntactically express pre-expectation $\wp{\stmt}(\expa)$.}
\begin{restatable}[Programmatic Rewards on Termination]{lemma}{lemmaRewardsOnTermination}\label{lem:rewards-on-termination}
    Let $\stmt \in \pGCL$ and $\expa \in \Expectations$. Then, $\wp{\stmt}(\expa) \eeq \wp{\stmt\symSemi \stmtTick{\expa}}(0) \eeq \Expected(\Rew{\stmt\symSemi \stmtTick{\expa}})$.
\end{restatable}
As the programmatic view indicates, this model only collects rewards on termination of the program.
In particular, diverging executions collect no rewards.

\subsection{Run-Times and Resource Consumption}\label{sec:run-times-and-resource-consumption}

The formal verification of expected run-times, and more generally resource consumption, is a well-studied topic in the literature~\cite{DBLP:conf/esop/KaminskiKMO16,DBLP:journals/jacm/KaminskiKMO18,DBLP:conf/tacas/KuraUH19,DBLP:books/cu/20/KaminskiKM20,DBLP:conf/pldi/Wang0R21}.
The analysis is subtle.
We already examined \Cref{fig:example-programs} which has a finite expected run-time despite having an infinite execution path (the database call never succeeding).
In the following, we examine \Cref{fig:approaches-runtime-tracking} which shows three different reward models for tracking runtimes for a schematic program consisting of a single while loop.

\paragraph{Unsoundness with counter variables.}
In \Cref{fig:rewards-termination}, the runtime is tracked in variable $\rt$ and the reward of $\rt$ is collected on termination of the program.
However, for e.g. the infinite loop ($\bexpr = \exprTrue$), the reward is zero, because the execution never reaches the $\stmtTick{\rt}$ statement.
This technique is unsound~\cite{DBLP:journals/jacm/KaminskiKMO18}.

An adaptation to track the second moment of the runtime is similarly deficient.
\Cref{fig:rewards-termination-squared} is a programmatic version of an approach of the $\mathsf{rt}$ calculus from~\cite{DBLP:conf/qest/KaminskiKM16}, developed to analyze covariances of runtimes.\footnote{The $\mathsf{rt}$ calculus from~\cite{DBLP:conf/qest/KaminskiKM16} is distinct from its successor, the $\symErt$ calculus~\cite{DBLP:conf/esop/KaminskiKMO16,DBLP:journals/jacm/KaminskiKMO18}.}
It collects the reward $\rt^2$ on termination.
But e.g. for the infinite loop, the reward is again zero.

A sufficient condition for the soundness of the counter variable approach (for the first moment of runtimes) was shown in~\cite{winkler_runtime_ast}: \emph{almost-sure termination} (AST), which is termination with probability $1$.
Proving AST is undecidable~\cite{DBLP:conf/mfcs/KaminskiK15}, and we therefore aim for an approach that avoids such a strong side-condition.

\begin{figure}[t]
    \centering
    \begin{subfigure}[c]{0.33\textwidth}
        \centering
        \begin{align*}
             & \stmtAsgn{\rt}{0}\symSemi             \\
             & \headerWhile{\bexpr}~\blockStart      \\
             & \quad \stmtAsgn{\rt}{\rt + 1}\symSemi \\
             & \quad \stmtP                          \\
             & \blockEnd                             \\
             & \stmtTick{\rt}
        \end{align*}
        \caption{Tracking runtimes via\\ rewards on termination,\\ corresponding to $\wp{\stmt}(\rt)$.}
        \label{fig:rewards-termination}
    \end{subfigure}%
    \begin{subfigure}[c]{0.33\textwidth}
        \centering
        \begin{align*}
             & \stmtAsgn{\rt}{0}\symSemi             \\
             & \headerWhile{\bexpr}~\blockStart      \\
             & \quad \stmtAsgn{\rt}{\rt + 1}\symSemi \\
             & \quad \stmtP                          \\
             & \blockEnd                             \\
             & \stmtTick{\rt^2}
        \end{align*}
        \caption{Tracking squared\\ runtimes via rewards on\\ termination, corresponding\\ to $\wp{\stmt}(\rt^2)$.}
        \label{fig:rewards-termination-squared}
    \end{subfigure}%
    \begin{subfigure}[c]{0.33\textwidth}
        \centering
        \begin{align*}
             & \phantom{\stmtAsgn{\rt}{0}\symSemi} \\
             & \headerWhile{\bexpr}~\blockStart    \\
             & \quad \stmtTick{1}\symSemi          \\
             & \quad \stmtP                        \\
             & \blockEnd                           \\
             & \phantom{\stmtTick{\rt^2}}
        \end{align*}
        \caption{Tracking runtimes via\\ incremental rewards,\\ corresponding to $\ert{\stmt}(0)$.}
        \label{fig:incremental-rewards}
    \end{subfigure}%

    \caption{Three different reward models for tracking runtimes for a schematic program consisting of a single while loop. The runtime of interest is the number of times the loop body $\stmtP$ is executed. Only \Cref{fig:incremental-rewards} tracks the runtime correctly, whereas \Cref{fig:rewards-termination,fig:rewards-termination-squared} have the result zero for a diverging loop ($\bexpr = \exprTrue$).}
    \label{fig:approaches-runtime-tracking}
\end{figure}

\paragraph{Incremental runtime collection.}
The $\symErt$ calculus~\cite{DBLP:conf/esop/KaminskiKMO16,DBLP:journals/jacm/KaminskiKMO18} is an extension of $\symWp$ to reason about expected runtimes for all -- including non-AST -- programs.\footnote{For simplicity, we consider the runtime model which counts loop iterations.}
The key idea is to collect the reward in each loop iteration incrementally.
Therefore, the reward is always incremented and diverging paths will collect infinite rewards.
To avoid using yet another calculus, we leverage programmatic reward modeling, as shown in \Cref{fig:incremental-rewards}.
The approach is sound without side-conditions on termination.

\begin{restatable}[Incremental Runtime Collection]{lemma}{lemmaIncrementalRuntimeTracking}\label{lem:incremental-runtime-tracking}
    Let loop-free $\stmtP \in \pGCL$ and $\bexpr \in \BExp$.
    Then, $\wp{\stmtWhile{\bexpr}{\stmtTick{1}\symSemi \stmtP}} \eeq \ert{\stmtWhile{\bexpr}{\stmtP}}$.
\end{restatable}
Thus, to correctly track the expected runtime of a program, one should avoid counter variables and instead use the $\symTick$ statement to collect the reward incrementally.
Extending this approach to other objectives such as higher moments of runtimes is the goal of \Cref{sec:reward-transformations}.

\subsection{Ghost Variables for Reward Modeling}

We now illustrate how using \emph{ghost variables}, we can encode a variety of reward structures via the $\symTick$ statement.
Ghost variables are program variables that do not affect the program's functional behavior, but are used to aid verification.
The term \emph{ghost variable} is standard in deductive verification; they are also called auxiliary (or observation) variables.
The counter variable approach presented just before is one such example.
We present several gadgets for different quantitative objectives, indicated as program transformations $\stmt \rightsquigarrow \stmtP$.
Assuming that ghost variables are initialized to $0$; the result of $\wp{\stmtP}(0)$ is the desired reward.

\paragraph{Discounted Rewards.}\label{sec:discounted-rewards}
In some applications, future rewards are considered less valuable than immediate rewards.
This is known as \emph{discounting}~\cite{puterman-discounting}.
We encode discounting with a ghost variable $\rt$ to track program steps, collecting rewards discounted by a factor $\gamma \in [0,1]$ at each time step.
\[
    \stmtTick{\aexpr} \quad\rightsquigarrow\quad \stmtTick{\gamma^\rt \cdot \aexpr}\symSemi \stmtAsgn{\rt}{\rt + 1}~.
\]

\paragraph{Step-Indexed Expected Values.}\label{sec:step-indexed-expected-values}
We obtain \emph{step-indexed} expected values, considered by e.g. \cite{DBLP:journals/pacmpl/MoosbruggerSBK22}, using a ghost variable $\rt$ to track program steps and a logical variable $N$ for the step of interest:
\[
    \stmtTick{\aexpr} \quad\rightsquigarrow\quad \stmtIfStart{\rt = N}~\stmtTick{\aexpr}~\blockEnd\symSemi \stmtAsgn{\rt}{\rt + 1}~.
\]
When instead condition $\rt \leq N$ is used, we track cumulative expected values up to step $N$.
This allows reasoning about \emph{long-run average rewards}, i.e. the limit of the average reward up to step $N$ as $N$ goes to infinity: $\limsup_{N \to \infty} \nicefrac{1}{N} \cdot \wp{\stmt}(0)$.

\paragraph{Expected Visiting Times.}\label{sec:expected-visiting-times}
The $\symWp$ calculus allows symbolic conditions $\bexpr \in \BExp$.
This enables reasoning about \emph{expected visiting times} (a.k.a. \emph{occupation measures}), i.e. the expected number of times states satisfying $\bexpr$ are visited during program execution.
\[
    \stmt \quad\rightsquigarrow\quad \stmt\symSemi \stmtIfStart{\bexpr}~\stmtTick{1}~\blockEnd
\]
Once computed, further expected values of random variable $X$ reduce to simple multiplication of $X$ by the expected visiting time~\cite{DBLP:journals/siamcomp/SharirPH84}.
This has been successfully used in probabilistic model checking~\cite{DBLP:journals/jar/MertensKQW25}.

\paragraph{First Visit and Return Times.}\label{sec:first-visit-times}\label{sec:first-return-times}
Similarly, gadgets to track the first visit time and first return time of a state satisfying condition $\bexpr \in \BExp$ can be built.
The ghost variable $\phi$ stores whether this is the first visit/return.
\begin{itemize}
    \item \emph{First Visit:}\quad $\stmt \:\rightsquigarrow\: \stmt\symSemi \stmtIfStart{\bexpr}~ \stmtTick{[\phi = 0]}\symSemi \stmtAsgn{\phi}{1}~\blockEnd$
    \item \emph{First Return:}\quad $\stmt \:\rightsquigarrow\: \stmt\symSemi \stmtIfStart{\bexpr}~ \stmtTick{[\phi = 1]}\symSemi \stmtAsgn{\phi}{\min(\phi + 1, 2)}~\blockEnd$
\end{itemize}

%% file: sections/4_encoding_moments_in_rewards.tex
\section{Reward Transformations}\label{sec:reward-transformations}

We now consider \emph{$f$-transformed cumulative rewards}, i.e. the expected value of some function $f$ applied to the cumulative reward.
Whereas it is simple to apply $f$ to expected rewards \emph{on termination} using the weakest pre-expectation calculus (\Cref{sec:rewards-on-termination}), this is not the case for cumulative expected rewards (\Cref{sec:run-times-and-resource-consumption}).
Instead of defining yet another variation of $\symWp$, we propose a program transformation $\rwTrans{f}{\stmt}$ for monotone functions $f \colon \PosRealsInf \to \PosRealsInf$.
The program $\stmt$ with expected reward $\Expected(\Rew{\stmt})$ is transformed to program $\rwTrans{f}{\stmt}$ whose expected reward is the $f$-transformed expected reward, i.e. $\Expected(\Rew{\rwTrans{f}{\stmt}}) = \Expected(f(\Rew{\stmt}))$.

\paragraph{Incremental reward updates.}
At the heart of the reward transformation is the idea of collecting $f$-transformed rewards incrementally, an idea inspired by the $\symAert$ calculus~\cite{DBLP:journals/pacmpl/BatzKKMV23}.
A $\stmtTick{\aexpr}$ statement simply adds $\aexpr$ to the cumulative reward; if the current cumulative reward is $\rt$, the new cumulative reward is $\rt + \aexpr$.
For $f$-transformed cumulative rewards, the collected reward equals $f(\rt)$ and the new cumulative reward equals $f(\rt {+} \aexpr)$.
The key idea is to (1) \emph{subtract} the previous cumulative reward $f(\rt)$ and (2) \emph{add} the new cumulative reward $f(\rt {+} \aexpr)$.

Consider the modeling of expected runtimes, where $1$ reward is collected at each step.
Then, the sequence of rewards that starts with $f(\rt)$ has the telescoping property and collapses to the final cumulative reward $f(\rt + N)$:\footnote{We define $\infty - x = \infty$ for all $x \in \PosRealsInf$ (in particular $\infty - \infty = \infty$) so that the sum is infinite if any $f(\rt + i)$ is infinite.}
\begin{align}\label{eq:telescoping}
     & \underbrace{f(\rt)}_\text{(0) previous} + \sum_{i=0}^{N-1} \underbrace{- f(\rt + i)}_\text{(1) reset} + \underbrace{f(\rt + i + 1)}_\text{(2) set}  \eeq f(\rt + N)~.
\end{align}
Since $f$ is monotone, the difference $f(\rt + \aexpr) - f(\rt)$ is non-negative, and we can encode this change in the program via a $\symTick$ statement.

\subsection{Program Transformation}\label{sec:program-transformation}

The program transformation is simple: we introduce a ghost variable $\rt$ to track the cumulative reward of the original program $\stmt$ and transform every $\stmtTick{\aexpr}$ statement to collect the difference $f(\rt + \aexpr) - f(\rt)$.
In the following, we denote by $f \colon \PosRealsInf \to \PosRealsInf$ a monotone function to ensure non-negative reward differences.

\begin{definition}[Reward Program Transformation]\label{def:reward-program-transformation}
    Let $\stmt \in \pGCL$ and $\rt$ be a fresh variable not occurring in $\stmt$.
    The \emph{$f$-transformed program} $\rwTrans{f}{\stmt}$ is
    \[
        \rwTrans{f}{\stmt} \eeq \stmtAsgn{\rt}{0}\symSemi \stmtTick{f(0)}\symSemi \rwTransP{f}{\stmt}~,
    \]
    where the auxiliary transformation $\rwTransP{f}{\stmt}$ replaces $\symTick$ statements:
    $$\rwTransP{f}{\stmtTick{\aexpr}} \eeq \stmtTick{(f(\rt + \aexpr) - f(\rt))}\symSemi \stmtAsgn{\rt}{\rt + \aexpr}~,$$
    and other statements are left unchanged.
\end{definition}
\Cref{fig:example-programs-transformed} shows the transformations of the example programs from \Cref{fig:example-programs} using $f(x) = x^2$, where the highlighted statements originate from the transformation.
For example, the $\stmtTick{1}$ statement in \Cref{fig:example-program-a} is transformed to $\stmtTick{( (\rt + 1)^2 - \rt^2 )}\symSemi \stmtAsgn{\rt}{\rt + 1}$ in \Cref{fig:example-program-a-transformed}.

\begin{figure}[t]
    \centering
    \begin{subfigure}[t]{0.45\textwidth}
        \centering
        \begin{align*}
             & \highlight{\stmtAsgn{\rt}{0}\symSemi}                                      \\
             & \highlight{\stmtReward{0}\symSemi}                                         \\
             & \stmtAsgn{done}{\exprFalse}\symSemi                            \\
             & \symWhile~(\neg done)~\blockStart                              \\
             & \quad \highlight{\stmtReward{2 \cdot \rt + 1}\symSemi}                     \\
             & \quad \highlight{\stmtAsgn{\rt}{\rt + 1}\symSemi}                          \\
             & \quad \stmtProb{\nicefrac{1}{2}}{\stmtAsgn{done}{\exprTrue}}{\stmtSkip} \\
             & \blockEnd
        \end{align*}
        \caption{Transformed version of \Cref{fig:example-program-a}.}
        \label{fig:example-program-a-transformed}
    \end{subfigure}
    \hfill
    \begin{subfigure}[t]{0.45\textwidth}
        \centering
        \begin{align*}
             & \highlight{\stmtAsgn{\rt}{0}\symSemi}                                      \\
             & \highlight{\stmtReward{\nicefrac{1}{4}}\symSemi}                           \\
             & \stmtAsgn{done}{\exprFalse}\symSemi                            \\
             & \symWhile~(\neg done)~\blockStart                              \\
             & \quad \highlight{\stmtReward{2 \cdot \tau + 1}\symSemi}                    \\
             & \quad \highlight{\stmtAsgn{\rt}{\rt + 1}\symSemi}                          \\
             & \quad \stmtProb{\nicefrac{2}{3}}{\stmtAsgn{done}{\exprTrue}}{\stmtSkip} \\
             & \blockEnd
        \end{align*}
        \caption{Transformed version of \Cref{fig:example-program-b}.}
        \label{fig:example-program-b-transformed}
    \end{subfigure}
    \caption{Transformations $\rwTrans{f}{\stmt}$ of programs $\stmt$ in \Cref{fig:example-programs} with $f(x) = x^2$.}
    \label{fig:example-programs-transformed}
\end{figure}

The expected reward of the transformed program $\rwTrans{f}{\stmt}$ is equal to the expected value of $f$ applied to the reward of the original program $\stmt$.
\begin{restatable}[Soundness]{theorem}{thmCorrectnessRewardTranslation}\label{thm:correctness-reward-translation}
    For $\stmt \in \pGCL$ and monotone $f \colon \PosRealsInf \to \PosRealsInf$: $$\Expected(\Rew{\rwTrans{f}{\stmt}}) = \Expected(f(\Rew{\stmt}))~.$$
\end{restatable}
In terms of weakest pre-expectation semantics, we can express \Cref{thm:correctness-reward-translation} as $\Expected(\Rew{\rwTrans{f}{\stmt}}) = \wp{\rwTrans{f}{\stmt}}(0) = \Expected(f(\Rew{\stmt}))$.
In particular, the $\symWp$ semantics of the (auxiliary) transformation of a $\symTick$ statement reflects the telescoping idea from \Cref{eq:telescoping}.
For an expectation $\expa \in \Expectations$, we have:
\[
    \wp{\rwTransP{f}{\stmtTick{\aexpr}}}(\expa) \morespace{=} \lam{\State}{\underbrace{f(\State(\rt) + \aexpr(\State))}_\text{(2) set} - \underbrace{f(\State(\rt))}_\text{(1) reset} + \underbrace{\expa(\State\substBy{\rt}{\rt + \aexpr})}_\text{next}}~.
\]
The program transformation preserves healthiness conditions similar to $\symWp$.
\begin{restatable}[Monotonicity]{theorem}{theoremMonotonicTransformations}\label{thm:monotonic-transformations}
    Let $\stmt \in \pGCL$ and let $f, g \colon \PosRealsInf \to \PosRealsInf$ be monotone with $f \expleq g$.
    For all $\expa \in \Expectations$, we have $\wp{\rwTrans{f}{\stmt}}(\expa) \expleq \wp{\rwTrans{g}{\stmt}}(\expa)$.
\end{restatable}
\noindent%
The transformation is linear with respect to scaling and shifting.
\begin{restatable}[Linearity]{theorem}{theoremLinearity}
    Let $\stmt \in \pGCL$ and let $f \colon \PosRealsInf \to \PosRealsInf$ be monotone.
    Then, for all $\alpha,\beta \in \PosReals$, we have $\wp{\rwTrans{\alpha \cdot f + \beta}{\stmt}}(0) = \alpha \cdot \wp{\rwTrans{f}{\stmt}}(0) + \beta$.
\end{restatable}
\noindent%
For linear transformation functions, we can avoid the ghost variable $\rt$ from the transformed program $\rwTrans{f}{\stmt}$ entirely.
\begin{restatable}[Ghost Buster]{theorem}{theoremGhostBuster}\label{theorem:ghost-buster}
    Let $f \colon \PosRealsInf \to \PosRealsInf$ be of the form $f(x) = \alpha \cdot x + \beta$ for some $\alpha,\beta \in \PosReals$.
    Then,
    \begin{align*}
        &\rwTrans{f}{\stmt} \eeq \stmtTick{\beta}\symSemi \rwTransP{f}{\stmt}
        \qquad\text{and}\qquad
        \rwTransP{f}{\stmtTick{\aexpr}} \eeq \stmtTick{\alpha \cdot \aexpr}~.
    \end{align*}
\end{restatable}
\noindent%
The transformation is well-behaved with respect to composition of transformation functions.
In essence: $\rwTransP{g}{\rwTransP{f}{\stmt}} = \rwTransP{g \circ f}{\stmt}$, where the equality is to be understood modulo the variables introduced by the transformations.

\begin{restatable}[Composition of Transformations]{theorem}{theoremCompositionTransformations}\label{theorem:composition-transformations}
    Let $\stmt \in \pGCL$ and let $f, g \colon \PosRealsInf \to \PosRealsInf$ be monotone.
    Assume that $\rwTransP{f}{\cdot}$ and $\rwTransP{g \circ f}{\cdot}$ use the same variable $\rt$ and $\rwTransP{g}{\cdot}$ uses a different variable $\rt'$.
    Then,
    \[
        \wp{\stmtAsgn{\rt'}{f(\rt)}\symSemi \rwTransP{g}{\rwTransP{f}{\stmt}}} = \wp{\rwTransP{g \circ f}{\stmt}\symSemi \stmtAsgn{\rt'}{f(\rt)}}~.
    \]
\end{restatable}

\subsection{Transforming Rewards of Markov Chains}\label{sec:mc-transformation}

To prove soundness of the program transformation (\Cref{thm:correctness-reward-translation}), we define a reward transformation for Markov chains directly.
The state space is augmented with a dimension $\PosRealsInf$ to track the accumulated reward of the original Markov chain, which is zero in the initial state $\State_1'$, and the reward function is adjusted to collect the incremental differences.
Here, we restrict to monotone $f$ with $f(0) = 0$ for simplicity.

\begin{restatable}[MC Transformation]{definition}{defiMcTransformation}\label{thm:defi-mc-transformation}
    Let MC $\mathcal{M} = (\States, \Prob, \State_1, \symRew)$ and monotone $f \colon \PosRealsInf \to \PosRealsInf$ with $f(0) = 0$.\footnote{The restriction $f(0) = 0$ in $\rwTrans{f}{\mathcal{M}}$ can be easily lifted by using the function $g(x) = f(x) - f(0)$ for $\rwTrans{g}{\mathcal{M}}$, and prepending a new initial state $\State_0'$ with reward $f(0)$. The full construction is presented in \appRefOrExt{sec:proofs-reward-transformations}.} \\
    We define $\rwTrans{f}{\mathcal{M}} = (\States \times \PosRealsInf,~ \ProbP,~ \State_1',~ \symRewP )$ where
    \begin{itemize}
        \item $\State_1' = (\State_1, 0)$,
        \item $\ProbP((\State, \alpha), (\State', \alpha')) = \begin{cases}
                      \Prob(\State)(\State') & \text{if } \alpha + \Rew{\State} = \alpha', \\
                      0                      & \text{otherwise}
                  \end{cases}$
        \item $\symRewP((\State, \alpha)) = f(\alpha + \Rew{\State}) - f(\alpha)$.
    \end{itemize}
\end{restatable}
\noindent%
Note that the transformation does not preserve the size of the reachable state space, because the added dimension $\PosRealsInf$ tracks the reward accumulation.

\begin{example}[Reachable States Explosion]\label{ex:mc-transformation-infinite}
    Let MC $\mathcal{M} = (\States, \Prob, \State_1, \symRew)$ with $\States = \Set{\State_1}$, $\Prob(\State_1)(\State_1) = 1$, $\symRew(\State_1) = 1$.
    Let $f(x) = x$.
    The transformed MC $\rwTrans{f}{\mathcal{M}}$ has an infinite path $(\State_1, 0), (\State_1, 1), (\State_1, 2), \ldots$ with probability one.
    The set of reachable states is $\Set{ (\State_1, \alpha) \mid \alpha \in \Nats }$, which is infinite.
\end{example}
\noindent%
The paths of the transformed Markov chain $\rwTrans{f}{\mathcal{M}}$ correspond one-to-one to the paths of the original $\mathcal{M}$, from which the soundness theorem follows.
\begin{restatable}{lemma}{lemmaMcTransformationSound}\label{lemma:mc-transformation-sound}
    Let MC $\mathcal{M} = (\States, \Prob, \State_1, \symRew)$ and $f \colon \PosRealsInf \to \PosRealsInf$ be monotone with $f(0) = 0$.
    Let $n \in \Nats$.
    For every $\State_1 \ldots \State_n \in \PathsN{\mathcal{M}}{=n}$, there is a one-to-one mapping to $(\State_1,\alpha_1)\ldots(\State_n,\alpha_n) \in \PathsN{\rwTrans{f}{\mathcal{M}}}{=n}$ in $\rwTrans{f}{\mathcal{M}}$ such that:
    \begin{enumerate}
        \item same path probabilities: $\Prob(\State_1 \ldots \State_n) = \ProbP((\State_1,\alpha_1) \ldots (\State_n,\alpha_n))$,
        \item $\alpha_i$ track previous untransformed rewards: $\forall 1 \leq i \leq n.\, \sum_{j < i} \Rew{\State_j} = \alpha_i$,
        \item transformed rewards: $f(\Rew{\State_1 \ldots \State_n}) = \RewP{(\State_1,\alpha_1) \ldots (\State_n,\alpha_n)}$.
    \end{enumerate}
\end{restatable}
\noindent%
It follows that the expected $f$-transformed reward of the original Markov chain $\mathcal{M}$ equals the expected reward of the transformed Markov chain $\rwTrans{f}{\mathcal{M}}$.

\begin{restatable}[Soundness]{theorem}{thmMcTransformationSound}\label{thm:mc-transformation-sound}
    For each MC $\mathcal{M}$ and monotone $f$: $$\Expected(f(\Rew{\mathcal{M}})) = \Expected(\RewP{\rwTrans{f}{\mathcal{M}}})~.$$
\end{restatable}
\noindent%
Note that one can lift the monotonicity requirement of function $f$.
For the soundness of the transformation, it suffices that $f(\rt + \aexpr) - f(\rt) \geq 0$ holds for all actually occurring values $\rt$ and $\aexpr$ in the Markov chain.

%% file: sections/5_transformation_functions.tex
\section{Reward Transformation Functions}\label{sec:reward-transformation-functions}

We apply the program transformation with different transformation functions $f$.

\subsection{Higher Moments of Rewards}\label{sec:higher-moments-of-rewards}

The \emph{$k$-th moment} of random variable $X$ is the expected value $\Expected(X^k)$.
Choosing the transformation function $f(x) = x^k$ for some $k \in \Nats_{> 0}$ yields higher moments of cumulative rewards.
By \Cref{thm:correctness-reward-translation}, we have $\Expected(\Rew{\stmt}^k) = \wp{\rwTrans{f}{\stmt}}(0)$.

\begin{example}[Variance of Web Server Runtime]\label{ex:variance-webserver}
    We return to \Cref{fig:example-programs-transformed}, which shows the transformations of the web server examples from \Cref{fig:example-programs} using $f(x) = x^2$, to obtain the second moment of the expected runtime.
    To calculate the expected reward, i.e. $\wp{\stmt}(0)$ of \Cref{fig:example-program-a-transformed}, we first tackle the loop.
    The $\symWp$ of the loop w.r.t. post $0$ is given by
    \[
        \lfp{\expa}{\iverson{\neg done} \cdot ((2 \cdot \rt + 1) + 0.5 \cdot \expa\substBy{done}{\exprTrue} + 0.5 \cdot \expa) + \iverson{done} \cdot 0}.
    \]
    It evaluates to $\iverson{\neg done} \cdot (4 \cdot \tau + 6)$.
    Applying $\symWp$ to the remaining first three lines yields $\wp{\rwTrans{f}{\stmt}} = 6$.
    Thus, the second moment of the expected reward of \Cref{fig:example-program-a} is $6$.
    Since the first moment is $2$ (c.f. \Cref{sec:wp-semantics}), the variance is $\mathrm{Var}(\Rew{\stmt}) = \Expected(\Rew{\stmt}^2) - (\Expected(\Rew{\stmt}))^2 = 6 - 2^2 = 2$.

    Similarly, one can show that $\wp{\stmt_\text{loop}}(0)$ of the loop in \Cref{fig:example-program-b-transformed} is equal to $\iverson{\neg done} \cdot (3 \cdot (\rt - 0.5) + 4.5)$.
    Applying $\symWp$ to the first three lines, we obtain $\wp{\rwTrans{f}{\stmt}}(0) = 4.75$, the second moment of the expected reward of \Cref{fig:example-program-b}.
    The variance of the expected reward is $\mathrm{Var}(\Rew{\stmt}) = 4.75 - 2^2 = 0.75$.
\end{example}

\begin{example}[Biased Random Walk]\label{ex:biased-random-walk-second-moment}
    Consider a random walk starting at position $x > 1$ that moves two steps to the left with probability $0.75$ and two steps to the right with probability $0.25$ until it reaches position $x \leq 1$.
    The second moment is upper-bounded by $x^2 + 3 \cdot x$, which we verified using the program transformation and $\symWp$ reasoning.
\end{example}

\subsection{Cumulative Distribution Functions}\label{sec:cdf}

We can use $f(x) = \iverson{x \geq N}$ for some $N \in \Nats$ to compute the cumulative distribution function (CDF) of the expected rewards.

\begin{example}[CDF of Web Server Runtime]\label{ex:cdf-webserver}
    To obtain the CDF for \Cref{fig:example-program-a}, we transform it with $f(x) = \iverson{x \geq N}$:
    \begin{align*}
         & \highlight{\stmtAsgn{\rt}{0}\symSemi \stmtTick{\iverson{0 \geq N}}\symSemi}~ \stmtAsgn{done}{\exprFalse} \symSemi                                                                                           \\
         & \headerWhile{\neg done}~\blockStart                                                                                                                      \\
         & \quad \highlight{\stmtTick{[\tau + 1 \geq N] - [\tau \geq N]}\symSemi \stmtAsgn{\rt}{\rt + 1}\symSemi}~ \stmtProb{\nicefrac{1}{2}}{\stmtAsgn{done}{\exprTrue}}{\stmtSkip} ~ \blockEnd
    \end{align*}
    Note that one may simplify the $\symTick$ statement to $\stmtTick{[\tau + 1 = N]}$.
    For the CDF, we get $\wp{\rwTrans{f}{\stmt}}(0) = 0.5^{N \monus 1}$, where $a \monus b = \max(0, a-b)$.
\end{example}

\subsection{Expected Excess}\label{sec:expected-excess}

Consider a fixed budget $N \in \PosReals$ and the objective is to determine how much the cumulative reward exceeds this budget in expectation.
The following example shows that we can use our program transformation to compute the \emph{expected excess} of a loop's runtime, i.e. $\Expected((\Rew{\stmt} \monus N))$ for all $N \in \PosReals$.

\begin{example}[Loop Splitting for Expected Excess]\label{ex:loop-splitting-excess}
    We compute the expected excess of the web server program from \Cref{fig:example-program-a}, depending on a success probability $p \in [0,1]$ of a database call and a budget $N \in \Nats$.
    Applying the transformation with $f(x) = x \monus N$ yields, after simplification\footnote{We use that $(((\rt + 1) \monus N) - (\rt \monus N))$ is equal to $[\tau \geq N]$.}, the following:
    \begin{align*}
         & \highlight{\stmtAsgn{\rt}{0}\symSemi}~ \stmtAsgn{done}{\exprFalse} \symSemi                                                                                           \\
         & \stmtWhile{\neg done}{
            \highlight{\stmtTick{[\tau \geq N]} \symSemi \stmtAsgn{\rt}{\rt + 1} \symSemi}~
            \stmtProb{p}{\stmtAsgn{done}{\exprTrue}}{\stmtSkip}
        }
        \intertext{
            Now, it is not difficult to see that the loop runs in two phases.
            The first phase runs until the runtime is below $N$, collecting no reward, and the second phase collects the excess of the runtime.
            Splitting the loop in two loops yields:
        }
         & \stmtAsgn{\rt}{0}\symSemi \stmtAsgn{done}{\exprFalse} \symSemi                                                                                           \\
         & \stmtWhile{\neg done \land (\rt < N)}{
            \stmtAsgn{\rt}{\rt + 1} \symSemi
            \stmtProb{p}{\stmtAsgn{done}{\exprTrue}}{\stmtSkip}
        }\symSemi                                                                                                                                                   \\
         & \stmtWhile{\neg done}{
            \stmtTick{1} \symSemi
            \stmtAsgn{\rt}{\rt + 1} \symSemi
            \stmtProb{p}{\stmtAsgn{done}{\exprTrue}}{\stmtSkip}
        }
    \end{align*}
    The second loop has an expected reward of $\nicefrac{1}{p}$ if $done = \exprFalse$, otherwise $0$.
    Formally: $\wp{loop2}(0) = \iverson{\neg done} \cdot \nicefrac{1}{p}$.
    The probability that the first loop terminates with $done = \exprFalse$ is $(1-p)^N$.
    Formally: $\wp{\stmt}(0) = \wp{loop1}(\wp{loop2}(0)) = \nicefrac{(1-p)^N}{p}$, which is the expected excess of the runtime over threshold $N$.
    For example, if $p = 0.1$ and $N = 10$, the expected excess is $\nicefrac{(1-0.1)^{10}}{0.1} \approx 3.486$.
\end{example}

\subsection{Moment-Generating Functions}\label{sec:mgf}

A \emph{moment-generating function} of a random variable $X$ is given by $\Expected(e^{t \cdot X})$ for $t \in \mathbb{R}_{\geq 0}$~\cite{Fristedt1996-cq}.
They can be used to obtain \emph{all} moments of $X$ by differentiating the MGF at $t = 0$.
With $f(x) = e^{t \cdot x}$, our reward transformation yields MGFs.\footnote{MGFs are similar to \emph{probability-generating functions} (PGFs), recently used for inference in probabilistic programs~\cite{DBLP:conf/cav/ChenKKW22,DBLP:conf/nips/ZaiserMO23,10.1145/3747534}. However, $f(x) = X^x$ for PGFs is not monotone for $X < 1$ and our transformation is not applicable.}

\begin{example}[MGF of a Randomized Assignment]\label{ex:mgf-coin-flip}
    Consider the program $\stmt = \stmtProb{p}{\stmtAsgn{x}{1}}{\stmtAsgn{x}{0}}\symSemi \stmtTick{x}$.
    Let $f(x) = e^{t \cdot x}$.
    One can show $\wp{\rwTrans{f}{\stmt}}(0) = (p \cdot (e^{t \cdot (\rt + 1)} - e^{t \cdot \rt}) + e^{t \cdot \rt})\substBy{\rt}{0} = p \cdot e^t + (1-p)$.
    The $n$-th derivative at $t = 0$ yields the $n$-th moment: $\Expected(X^n) = p$ for all $n \in \Nats_{>0}$.
\end{example}

\subsection{Multiple Rewards}\label{sec:multiple-rewards}

Sometimes, it is useful to track multiple rewards simultaneously and aggregate them into a single reward.
In some cloud computing platforms, for example, pricing is based on the product of the \emph{runtime} and \emph{memory consumption}~\cite{aws-lambda-pricing}.

Our transformation can be generalized to a multi-reward setting by introducing multiple reward variables $\rt_i$.
The generalized reward statements assign multiple rewards at once, e.g. $\stmtTick{(\aexpr_1, \aexpr_2)}$ assigns reward $\aexpr_1$ to reward $1$ and reward $\aexpr_2$ to reward $2$.
The transformation $\rwTransP{f}{\stmt}$ is generalized to $f \colon (\PosRealsInf)^n \to \PosRealsInf$ that maps $n$ rewards to a single reward:
\begin{align*}
    \rwTransP{f}{\stmtTick{(\aexpr_1,\ldots,\aexpr_n)}} =~& \stmtTick{f(\rt_1+\aexpr_1,\ldots,\rt_n+\aexpr_n)-f(\rt_1,\ldots,\rt_n)}\symSemi \\
    &\stmtAsgn{\rt_1}{\rt_1+\aexpr_1} \symSemi \ldots\symSemi \stmtAsgn{\rt_n}{\rt_n+\aexpr_n}
\end{align*}

\begin{example}[Expected Cost with Runtime and Memory Consumption]\label{ex:multiple-rewards-cost}
    Consider a program starting with two units of memory and computation each.
    With probability $p$ it needs more computation, for this, it requests an additional unit of memory with probability $q$.
    We can write this as a program, apply our transformation with $f(\rt_1, \rt_2) = \rt_1 \cdot \rt_2$ (modeling cost as product of runtime and memory), and compute the expected cost.
    The result is $\wp{\rwTrans{f}{\stmt}}(0) = 4 + (p \cdot (2 + (q \cdot 3)))$.
\end{example}

%% file: sections/6_automation.tex
\section{Automation}\label{sec:automation}

The off-the-shelf \emph{Caesar} deductive verifier for probabilistic programs~\cite{DBLP:journals/pacmpl/SchroerBKKM23} was used, which supports $\symTick$ statements and implements various proof rules for reasoning about lower and upper bounds of expected rewards.
One example is the \emph{induction proof rule}~\cite{parkFixpointInductionProofs1969,DBLP:phd/dnb/Kaminski19}, to verify \emph{upper bounds} $\colheylo{I} \in \expa$ on $\wp{\stmt}(\expa)$ for $\stmt = \stmtWhile{\bexpr}{\stmt}$ and $\expa \in \Expectations$, where $\colheylo{I}$ is a user-provided loop invariant:
\[
    \iverson{\bexpr} \cdot \wp{\stmt}(\colheylo{I}) + \iverson{\neg \bexpr} \cdot \expa \sqsubseteq \colheylo{I} \quad\text{implies}\quad \wp{\stmtWhile{\bexpr}{\stmt}}(\expa) \sqsubseteq \colheylo{I}~.
\]
We provide Caesar input files for examples in \Cref{sec:reward-transformation-functions} in \appRefOrExt{sec:case-studies}, verifying upper bounds on expected rewards using the induction proof rule.
In addition, Caesar supports analysis of (finite) operational semantics with integer variables via the probabilistic model checker \emph{Storm}~\cite{DBLP:journals/sttt/HenselJKQV22}.
For infinite-state operational semantics (c.f. \Cref{fig:webserver-markov-chain}), Storm can under-approximate the expected value.

%% file: sections/7_related_work.tex
\section{Related Work}\label{sec:related-work}

\paragraph{Expected Run-Time Analysis.}
In \Cref{sec:run-times-and-resource-consumption}, we discussed the lack of support for e.g. higher moments of the $\symErt$ calculus~\cite{DBLP:conf/esop/KaminskiKMO16,DBLP:journals/jacm/KaminskiKMO18}.
The same problems apply to the \emph{amortized expected runtime calculus}~\cite{DBLP:journals/pacmpl/BatzKKMV23}, which extends the $\symErt$ calculus to reason about potential functions, see \appRefOrExt{sec:more-related-work} for details.
Automated inference of $\symErt$-style run-times has been studied e.g. in~\cite{DBLP:conf/pldi/NgoC018}.
We assume non-negative rewards; analyses for mixed-sign rewards include \cite{wang2019cost,chatterjee2024}.
As a complimentary approach, \cite{asymptotic-analysis} propose asymptotic bounds for run-time analysis instead of expected values.

\paragraph{Higher Moments.}
\Citeauthor*{DBLP:conf/qest/KaminskiKM16} aim to reason about covariances of probabilistic programs~\cite{DBLP:conf/qest/KaminskiKM16}.
They define a semantics called $\mathsf{rt}$ that increments a dedicated runtime variable $\tau$ by one for each statement executed and consider the expected value of $\tau^k$ on termination.
However, as explained in \Cref{sec:run-times-and-resource-consumption}, this approach is deficient for diverging programs.
\Citeauthor*{DBLP:conf/tacas/KuraUH19}~\cite{DBLP:conf/tacas/KuraUH19} focus on tail probabilities of run-times and make use of higher moments to obtain bounds on them.
They synthesize upper bounds on the $k$-th moment via template-based invariants.
In~\cite{DBLP:conf/pldi/Wang0R21}, \citeauthor*{DBLP:conf/pldi/Wang0R21} present an automated approach for the inference of higher moments of expected runtimes of probabilistic programs.
As for the previous approach, proof rules have to be proven specifically for their semantics.
Lower bounds are also considered and used to derive central moments.
Their approach is automated, and uses template-based invariant synthesis for interval bounds.
Algebraic methods for higher moments of step-indexed expected value have been used in~\cite{DBLP:journals/pacmpl/MoosbruggerSBK22}.
On programs with certain restrictions, e.g. acyclic non-linear variable dependencies and variables in conditions only assuming finitely many values, they can compute closed-form solutions for all moments.
Hardness results for higher moments of step-indexed expected values are presented in~\cite{DBLP:journals/pacmpl/MullnerMK24}.

%% file: sections/8_conclusion.tex
\section{Conclusion}\label{sec:conclusion}

We presented probabilistic programs with reward statements to obtain a one-fits-all approach for analyzing multiple objectives, e.g. expected runtimes, expected visiting times, higher moments of expected rewards, or the expected excess over a threshold.
All reduce to the ``classical'' $\symWp$ calculus, eliminating the need for dedicated calculi/semantics for each objective.
Therefore, existing proof rules and tools for $\symWp$ can be reused in our setting.
Future work includes the incorporation of programs with non-determinism; programmatic reward modeling works in this setting as well, but the reward transformation introduces additional subtleties due to additional information introduced by the counter variable $\rt$.

%% file: appendices/omitted_proofs.tex
\section{Proofs}\label{sec:omitted_proofs}

\subsection[Definitions for Section~\ref{sec:foundations}]{Definitions for \Cref{sec:foundations}}\label{sec:definitions-foundations}

\begin{figure}[t]
    \centering
    \begin{minipage}{0.95\linewidth}
        \medskip

        \scriptsize{}
        \begin{align*}
             &
            \infer%
            {(\term, \State) \trans{1} \final}
            {}
            \qquad\qquad
            \infer%
            {\final \trans{1} \final}
            {}
            \\[1.5ex]
             &
            \rlap{
                \infer%
                {(\stmtSkip, \State) \trans{1} (\term, \State)}
                {}
                \qquad\quad
                \infer%
                {(\stmtAsgn{x}{\aexpr}, \State) \trans{1} (\term, \State\substBy{x}{\aexpr})}
                {}
                \qquad\quad
                \infer%
                {(\stmtTick{\aexpr}, \State) \trans{1} (\term, \State)}
                {}
            }
            \\[1.5ex]
             &
            \infer%
            {(\stmtOne \symSemi \stmtTwo, \State) \trans{p} (\stmtTwo, \State')}
            {(\stmtOne, \State) \trans{p} (\term, \State')}
             &   &
            \infer%
            {(\stmtOne \symSemi \stmtTwo, \State) \trans{p} (\stmtOne' \symSemi \stmtTwo, \State')}
            {(\stmtOne, \State) \trans{p} (\stmtOne', \State')}
            \\[1.5ex]
             &
            \infer%
            {(\stmtProb{p}{\stmtOne}{\stmtTwo}, \State) \trans{p} (\stmtOne, \State)}
            {}
             &   &
            \infer%
            {(\stmtProb{p}{\stmtOne}{\stmtTwo}, \State) \trans{1-p} (\stmtTwo, \State)}
            {}
            \\[1.5ex]
             &
            \infer%
            {(\stmtIf{\bexpr}{\stmtOne}{\stmtTwo}, \State) \trans{1} (\stmtOne, \State)}
            {\llbracket \bexpr \rrbracket_\State = \exprTrue}
             &   &
            \infer%
            {(\stmtIf{\bexpr}{\stmtOne}{\stmtTwo}, \State) \trans{1} (\stmtTwo, \State)}
            {\llbracket \bexpr \rrbracket_\State = \exprFalse}
            \\[1.5ex]
             &
            \infer%
            {(\stmtWhile{\bexpr}{\stmt}, \State) \trans{1} (\stmt \symSemi \stmtWhile{\bexpr}{\stmt}, \State)}
            {\llbracket \bexpr \rrbracket_\State = \exprTrue}
             &   &
            \infer%
            {(\stmtWhile{\bexpr}{\stmt}, \State) \trans{1} (\term, \State)}
            {\llbracket \bexpr \rrbracket_\State = \exprFalse}
        \end{align*}
    \end{minipage}
    \caption{Operational semantics for $\pGCL$ programs, where $(\stmtOne, \State) \trans{p} (\stmtTwo, \State')$ iff $\Prob((\stmtOne, \State))((\stmtTwo, \State')) = p$, for $\State \in \States$ and $\stmt, \stmtOne, \stmtTwo \in \pGCL$.}
    \label{fig:operational-semantics}
\end{figure}

For a program $\stmt$, its \emph{operational Markov chain} $\progMdp{\stmt}{\State_1} = (\Conf, \Prob, \conf_1, \symRew)$ is given by the reachable configurations $\Conf \subseteq (\pGCL \cup \Set{\term}) \times \States \cup \Set{\final}$ with transition probabilities $\Prob$ as the smallest relation defined by the inference rules in~\Cref{fig:operational-semantics}, initial configuration $\conf_1 = (\stmt, \State_1)$, and reward function $\symRew$ defined as follows:
\begin{align*}
    \Rew{(\stmt, \State)} = \begin{cases}
        r & \text{if } (\exists \stmt' \in \pGCL.~ \stmt = \stmtTick{r}\symSemi \stmt') \lor (\stmt = \stmtTick{r}) \\
        0 & \text{otherwise}
    \end{cases}
\end{align*}
For the web server program in \Cref{fig:example-program-a}, \Cref{fig:webserver-markov-chain-full} shows its operational Markov chain from an arbitrary initial state $\State$ with explicit configurations, and is a full version of \Cref{fig:webserver-markov-chain} from the main text.

\begin{figure}[t]
    \centering
    {\scriptsize
    \[
    \begin{aligned}
        B &\coloneqq \stmtTick{1}\symSemi \stmtProb{\nicefrac{1}{2}}{\stmtAsgn{done}{\exprTrue}}{\stmtSkip},
        \qquad
        L \coloneqq \stmtWhile{\neg done}{B},
        \\
        \State_F &\coloneqq \State\substBy{done}{\exprFalse},
        \qquad
        \State_T \coloneqq \State_F\substBy{done}{\exprTrue}
    \end{aligned}
    \]
    \begin{tikzpicture}[
        cfg/.style={draw, rounded corners, align=center, inner sep=1.3pt, font=\scriptsize},
        sink/.style={circle, draw, inner sep=1.6pt, font=\scriptsize},
        node distance=0.95cm and 1.1cm,
    ]
        \node[cfg, initial, initial where=left, initial text={}] (s1) {$\bigl(\stmtAsgn{done}{\exprFalse}\symSemi L,\State\bigr)$};
        \node[cfg, below=1.05cm of s1] (s2) {$\bigl(L,\State_F\bigr)$};
        \node[cfg, left=0.35cm of s2] (s3) {$\bigl(B\symSemi L,\State_F\bigr)$};
        \node[above=0.05cm of s3, font=\scriptsize] {$\Rew{(B\symSemi L,\State_F)} = 1$};
        \node[cfg, below left=0.55cm and 0.15cm of s3] (s4) {$\bigl(\stmtAsgn{done}{\exprTrue}\symSemi L,\State_F\bigr)$};
        \node[cfg, below right=0.55cm and 0.15cm of s4] (s5) {$\bigl(L,\State_T\bigr)$};
        \node[cfg, below right=0.55cm and 0.15cm of s3] (s6) {$\bigl(\stmtSkip\symSemi L,\State_F\bigr)$};
        \node[cfg, right=0.5cm of s5] (term) {$\bigl(\term,\State_T\bigr)$};
        \node[sink, right=0.5cm of term] (final) {$\bot$};

        \path[->] (s1) edge node {} (s2);
        \path[->] (s2) edge node {} (s3);
        \path[->] (s3) edge node[left] {\scriptsize{}0.5} (s4);
        \path[->] (s3) edge node[right] {\scriptsize{}0.5} (s6);
        \path[->] (s4) edge node {} (s5);
        \path[->] (s5) edge node {} (term);
        \path[->] (term) edge node {} (final);
        \path[->] (final) edge[loop right] node {} (final);
        \path[->,bend right=60] (s6) edge node {} (s2);
    \end{tikzpicture}}
    \caption{Configuration-level operational MC for the web server program from \Cref{fig:example-program-a}, instantiated from an arbitrary initial state $\State$; this is a full version of \Cref{fig:webserver-markov-chain}.}
    \label{fig:webserver-markov-chain-full}
\end{figure}
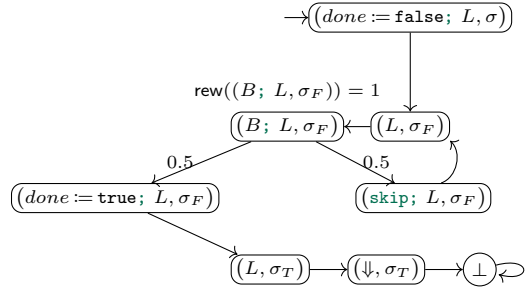

\subsection[Proofs for Section~\ref{sec:programmatic-reward-modeling}]{Proofs for \Cref{sec:programmatic-reward-modeling}}

\lemmaRewardsOnTermination*

\begin{proof}
    Let $\stmt \in \pGCL$, $\expa \in \Expectations$, and $\State \in \States$.
    By $\symWp$ semantics, we have $\wp{\stmt}(\expa)(\State) = \wp{\stmt}(\wp{\stmtTick{\expa}}(0))(\State)$.
    By $\symWp$ soundness, this is equal to $\wp{\stmt\symSemi \stmtTick{\expa}}(0)(\State)=\Expected(\RewStateInd{\State}{0}{\stmt\symSemi \stmtTick{\expa}}) = \Expected(\RewState{\State}{\stmt\symSemi \stmtTick{\expa}}).$
\end{proof}

\begin{figure}[t]
    \centering
    \renewcommand{\arraystretch}{1.2}%
    \setlength{\tabcolsep}{6pt}%
    \begin{tabular}{@{}l l@{}}
        \toprule
        $\stmt$ & $\wp{\stmt}(\expa)$ \\
        \midrule
        $\stmtSkip$ & $\expa$ \\
        $\stmtAsgn{x}{\aexpr}$ & $\expa\substBy{x}{\aexpr}$ \\
        $\stmtOne \symSemi \stmtTwo$ & $\wp{\stmtOne}(\wp{\stmtTwo}(\expa))$ \\
        $\stmtProb{p}{\stmtOne}{\stmtTwo}$ & $p \cdot \wp{\stmtOne}(\expa) + (1-p) \cdot \wp{\stmtTwo}(\expa)$ \\
        $\stmtIf{\bexpr}{\stmtOne}{\stmtTwo}$ & $\iverson{\bexpr} \cdot \wp{\stmtOne}(\expa) + \iverson{\neg\bexpr} \cdot \wp{\stmtTwo}(\expa)$ \\
        $\stmtWhile{\bexpr}{\stmtP}$ & $\lfp{\expb}{\iverson{\bexpr} \cdot (1+\wp{\stmtP}(\expb)) + \iverson{\neg \bexpr} \cdot \expa}$ \\
        \bottomrule
    \end{tabular}
    \caption{Inductive definition of the \emph{expected run-time transformer} for program $\stmt$, $\ert{\stmt} \colon \Expectations \to \Expectations$ where $\expa \in \Expectations$ is the \emph{post-runtime}. The only differences to standard $\symWp$ semantics (\Cref{fig:wp-semantics}) is the ``$+1$'' in the while loop and missing $\symTick$ statement.}
    \label{fig:ert-semantics}
\end{figure}

\lemmaIncrementalRuntimeTracking*

\begin{proof}
    \allowdisplaybreaks
    We consider the expected runtime semantics $\ert{\stmt}$ from~\cite{DBLP:conf/esop/KaminskiKMO16,DBLP:journals/jacm/KaminskiKMO18} with the runtime model that counts loop iterations only.
    \Cref{fig:ert-semantics} shows the inductive definition of $\ert{\stmt}$, which only differs from standard $\symWp$ semantics by the ``$+1$'' in the while loop case and the absence of the $\symTick$ statement.

    Consider a loop $\stmtWhile{\bexpr}{\stmtP}$.
    We show for all expectations $\expa \in \Expectations$:
    \[
        \wp{\stmtWhile{\bexpr}{\stmtTick{1}\symSemi \stmtP}}(\expa) \eeq \ert{\stmtWhile{\bexpr}{\stmtP}}(\expa)~.
    \]
    For the programmatic reward loop, we have:
    \begin{align*}
        & \wp{\stmtWhile{\bexpr}{\stmtTick{1}\symSemi \stmtP}}(\expa) \\
        &= \lfp{\expb}{\iverson{\bexpr} \cdot \wp{\stmtTick{1}\symSemi \stmtP}(\expb) + \iverson{\neg \bexpr} \cdot \expa} \\
        &= \lfp{\expb}{\iverson{\bexpr} \cdot \wp{\stmtTick{1}}(\wp{\stmtP}(\expb)) + \iverson{\neg \bexpr} \cdot \expa} \\
        &= \lfp{\expb}{\iverson{\bexpr} \cdot (1 + \wp{\stmtP}(\expb)) + \iverson{\neg \bexpr} \cdot \expa}~,
    \intertext{which is equal to the expected runtime semantics:}
        & \ert{\stmtWhile{\bexpr}{\stmtP}}(\expa)
    \end{align*}
    The last equality holds because $\ert{\stmtP} = \wp{\stmtP}$ since $\stmtP$ is loop-free.

    This lemma can be extended to nested loops by induction on the structure of the program, and adding the $\symTick$ statements into every loop body.
\end{proof}

\subsection[Proofs for Section~\ref{sec:reward-transformations}]{Proofs for \Cref{sec:reward-transformations}}\label{sec:proofs-reward-transformations}

We state the program transformation functions of \Cref{def:reward-program-transformation} in more detail here, to aid the proofs below.
\Cref{fig:reward-transformation} defines the program transformation $\rwTrans{f}{\stmt}$ and an auxiliary transformation $\rwTransP{f}{\stmt}$.
The (outer) transformation $\rwTrans{f}{\stmt}$ simply adds the variable $\rt$, initializes it to $0$, prepends $\stmtTick{f(0)}$, and invokes the auxiliary transformation $\rwTransP{f}{\stmt}$.
The latter handles the transformation of $\symTick$ statements, recursively applying the transformation to the sub-statements of $\stmt$.
Every $\stmtTick{\aexpr}$ statement is transformed to $\stmtTick{(f(\rt + \aexpr) - f(\rt))}$, and the variable $\rt$ is updated with $\stmtAsgn{\rt}{\rt + \aexpr}$.
Other statements are left unchanged.

\begin{figure}[t]
    \centering
    \renewcommand{\arraystretch}{1.3}
    \setlength{\tabcolsep}{1em} %
    \begin{tabular}{l l}
        \toprule
        $\stmt$                               & $\rwTrans{f}{\stmt}$                                                         \\
        \midrule
        $\stmt$                               & $\stmtAsgn{\rt}{0}\symSemi \stmtTick{f(0)}\symSemi \rwTransP{f}{\stmt}$      \\[3ex]
        \toprule
        $\stmt$                               & $\rwTransP{f}{\stmt}$                                                        \\
        \midrule
        $\stmtSkip$                           & $\stmtSkip$                                                                  \\
        $\stmtTick{\aexpr}$                   & $\stmtTick{(f(\rt + \aexpr) - f(\rt))}\symSemi \stmtAsgn{\rt}{\rt + \aexpr}$ \\
        $\stmtAsgn{x}{\aexpr}$                & $\stmtAsgn{x}{\aexpr}$                                                       \\
        $\stmtProb{p}{\stmtOne}{\stmtTwo}$    & $\stmtProb{p}{\rwTransP{f}{\stmtOne}}{\rwTransP{f}{\stmtTwo}}$               \\
        $\stmtIf{\bexpr}{\stmtOne}{\stmtTwo}$ & $\stmtIf{\bexpr}{\rwTransP{f}{\stmtOne}}{\rwTransP{f}{\stmtTwo}}$            \\
        $\stmtSeq{\stmtOne}{\stmtTwo}$        & $\rwTransP{f}{\stmtOne}\symSemi \rwTransP{f}{\stmtTwo}$                      \\
        $\stmtWhile{\bexpr}{\stmtP}$          & $\stmtWhile{\bexpr}{\rwTransP{f}{\stmtP}}$                                   \\
        \bottomrule
    \end{tabular}
    \caption{Program transformation $\rwTrans{f}{\stmt}$ and auxiliary transformation $\rwTransP{f}{\stmt}$ for monotonic $f \colon \PosRealsInf \to \PosRealsInf$ and fresh variable $\rt$ not occurring in $\stmt$.}
    \label{fig:reward-transformation}
\end{figure}

\Cref{thm:correctness-reward-translation} states that the transformation $\rwTrans{f}{\stmt}$ correctly transforms the expected reward of $\stmt$ according to $f$.
The proof of \Cref{thm:correctness-reward-translation} is mostly technical; the main conceptual argument is provided by \Cref{thm:mc-transformation-sound} (c.f. \Cref{sec:mc-transformation}).

\thmCorrectnessRewardTranslation*

\begin{proof}
    Let $\State_1 \in \States$ be an arbitrary initial state.
    Define
    \[
        \mathcal{M}_{\mathsf{prog}} \coloneqq \progMdp{\rwTrans{f}{\stmt}}{\State_1},
        \qquad
        \mathcal{M}_{\mathsf{mc}} \coloneqq \rwTrans{f}{\progMdp{\stmt}{\State_1}}~.
    \]
    We compare these two Markov chains, where $\mathcal{M}_{\mathsf{mc}}$ is obtained by the Markov chain transformation from \Cref{thm:mc-transformation-sound}.
    We need to assume that $f(0) = 0$ holds for \Cref{thm:mc-transformation-sound}, but \Cref{lem:mc-transformation-sound-lifted} shows that this restriction can be lifted by another simple transformation.
    The Markov chains $\mathcal{M}_{\mathsf{prog}}$ and $\mathcal{M}_{\mathsf{mc}}$ are essentially the same, except for the following differences:
    \begin{itemize}
        \item The states of $\mathcal{M}_{\mathsf{prog}}$ are program states $\State \in \States$ (extended with fresh variable $\rt$), while the states of $\mathcal{M}_{\mathsf{mc}}$ are pairs $(\State, \rt) \in \States \times \PosRealsInf$. These are equivalent representations of the same states.
        \item The statement $\stmtTick{\aexpr}$ is transformed into the sequential composition of $\stmtTick{(f(\rt + \aexpr) - f(\rt))}$ and $\stmtAsgn{\rt}{\rt + \aexpr}$.
              Let $\State_i$ be the state at some point $i$ of the execution of $\stmt$ and let $\alpha_i$ be the cumulative reward collected up to this point.
              At a $\symTick$ statement, $\mathcal{M}_{\mathsf{mc}}$ will execute one step:
              \[
                  ((\stmtTick{(f(\rt + \aexpr) - f(\rt))}, \State_i), \alpha_i)~.
              \]
              On the other hand, we have three steps in $\mathcal{M}_{\mathsf{prog}}$:
              \begin{align*}
                  &(\stmtTick{(f(\rt + \aexpr) - f(\rt))}\symSemi \stmtAsgn{\rt}{\rt + \aexpr}, \State_i') \\
                  &\quad \trans{1} (\term\symSemi \stmtAsgn{\rt}{\rt + \aexpr}, \State_i') \\
                  &\quad \trans{1} (\stmtAsgn{\rt}{\rt + \aexpr}, \State_i')~.
              \end{align*}
              However, the probabilities and cumulative rewards of these paths are equal.
    \end{itemize}
    As a result, we can conclude that the expected rewards of both Markov chains are equal:
    \begin{align*}
        \Expected(\Rew{\mathcal{M}_{\mathsf{prog}}}) & \eeq \Expected(\Rew{\mathcal{M}_{\mathsf{mc}}})~.
        \intertext{By \Cref{thm:mc-transformation-sound}, $\mathcal{M}_{\mathsf{mc}}$ satisfies}
        \Expected(\Rew{\mathcal{M}_{\mathsf{mc}}})    & \eeq \Expected(f(\Rew{\progMdp{\stmt}{\State_1}}))~.
    \end{align*}
    Thus, we have $\Expected(f(\RewState{\State_1}{\stmt})) = \Expected(f(\Rew{\progMdp{\stmt}{\State_1}})) = \Expected(\Rew{\progMdp{\rwTrans{f}{\stmt}}{\State_1}})$ for all $\State_1 \in \States$, which is equivalent to $\Expected(\Rew{\rwTrans{f}{\stmt}}) = \Expected(f(\Rew{\stmt}))$.
\end{proof}

\theoremMonotonicTransformations*

\begin{proof}
    \allowdisplaybreaks
    Let $\stmt \in \pGCL$ and let $f, g \colon \PosRealsInf \to \PosRealsInf$ be monotonic functions with $f(x) \leq g(x)$ for all $x \in \PosRealsInf$.
    We show by structural induction on $\stmt$ that
    \begin{align}
        \forall \expa \in \Expectations.\quad \wp{\rwTransP{f}{\stmt}}(\expa) \leq \wp{\rwTransP{g}{\stmt}}(\expa)~. \label{eq:monotonic-transformations-induction}
    \end{align}

    \noindent
    Base cases:

    $\bullet$ $\stmt = \stmtSkip$:
    \begin{align*}
         & \wp{\rwTransP{f}{\stmtSkip}}(\expa) = \wp{\stmtSkip}(\expa) = \expa \\
         & = \wp{\stmtSkip}(\expa) = \wp{\rwTransP{g}{\stmtSkip}}(\expa)
    \end{align*}

    $\bullet$ $\stmt = \stmtAsgn{x}{\aexpr}$:
    \begin{align*}
         & \wp{\rwTransP{f}{\stmtAsgn{x}{\aexpr}}}(\expa) = \wp{\stmtAsgn{x}{\aexpr}}(\expa) = \expa\substBy{x}{\aexpr} \\
         & = \wp{\stmtAsgn{x}{\aexpr}}(\expa) = \wp{\rwTransP{g}{\stmtAsgn{x}{\aexpr}}}(\expa)
    \end{align*}

    $\bullet$ $\stmt = \stmtTick{\aexpr}$:
    \begin{align*}
         & \wp{\rwTransP{f}{\stmtTick{\aexpr}}}(\expa)                                                                                                    \\
         & = \wp{\stmtTick{(f(\rt + \aexpr) - f(\rt))}\symSemi \stmtAsgn{\rt}{\rt + \aexpr}}(\expa) \tag{definition of $\rwTransP{f}{\stmtTick{\aexpr}}$} \\
         & = (f(\rt + \aexpr) - f(\rt)) + \expa\substBy{\rt}{\rt + \aexpr} \tag{definition of $\symWp$}                                                   \\
         & \leq (g(\rt + \aexpr) - g(\rt)) + \expa\substBy{\rt}{\rt + \aexpr} \tag{monotonicity of $f$ and $g$}                                           \\
         & = \wp{\stmtTick{(g(\rt + \aexpr) - g(\rt))}\symSemi \stmtAsgn{\rt}{\rt + \aexpr}}(\expa) \tag{definition of $\symWp$}                          \\
         & = \wp{\rwTransP{g}{\stmtTick{\aexpr}}}(\expa) \tag{definition of $\rwTransP{g}{\stmtTick{\aexpr}}$}
    \end{align*}

    \noindent
    Now assume the induction hypothesis holds for $\stmtOne, \stmtTwo, \stmtP \in \pGCL$.

    $\bullet$ $\stmt = \stmtSeq{\stmtOne}{\stmtTwo}$:
    \begin{align*}
         & \wp{\rwTransP{f}{\stmtSeq{\stmtOne}{\stmtTwo}}}(\expa)                                                        \\
         & = \wp{\rwTransP{f}{\stmtOne}\symSemi \rwTransP{f}{\stmtTwo}}(\expa) \tag{definition of $\rwTransP{f}{\cdot}$} \\
         & = \wp{\rwTransP{f}{\stmtOne}}(\wp{\rwTransP{f}{\stmtTwo}}(\expa)) \tag{definition of $\symWp$}                \\
         & \leq \wp{\rwTransP{g}{\stmtOne}}(\wp{\rwTransP{g}{\stmtTwo}}(\expa)) \tag{induction hypothesis}               \\
         & = \wp{\rwTransP{g}{\stmtOne}\symSemi \rwTransP{g}{\stmtTwo}}(\expa) \tag{definition of $\symWp$}              \\
         & = \wp{\rwTransP{g}{\stmtSeq{\stmtOne}{\stmtTwo}}}(\expa) \tag{definition of $\rwTransP{g}{\cdot}$}
    \end{align*}

    $\bullet$ $\stmt = \stmtProb{p}{\stmtOne}{\stmtTwo}$:
    \begin{align*}
         & \wp{\rwTransP{f}{\stmtProb{p}{\stmtOne}{\stmtTwo}}}(\expa)                                                                    \\
         & = \wp{\stmtProb{p}{\rwTransP{f}{\stmtOne}}{\rwTransP{f}{\stmtTwo}}}(\expa) \tag{definition of $\rwTransP{f}{\cdot}$}          \\
         & = p \cdot \wp{\rwTransP{f}{\stmtOne}}(\expa) + (1 - p) \cdot \wp{\rwTransP{f}{\stmtTwo}}(\expa) \tag{definition of $\symWp$}  \\
         & \leq p \cdot \wp{\rwTransP{g}{\stmtOne}}(\expa) + (1 - p) \cdot \wp{\rwTransP{g}{\stmtTwo}}(\expa) \tag{induction hypothesis} \\
         & = \wp{\stmtProb{p}{\rwTransP{g}{\stmtOne}}{\rwTransP{g}{\stmtTwo}}}(\expa) \tag{definition of $\symWp$}                       \\
         & = \wp{\rwTransP{g}{\stmtProb{p}{\stmtOne}{\stmtTwo}}}(\expa) \tag{definition of $\rwTransP{g}{\cdot}$}
    \end{align*}

    $\bullet$ $\stmt = \stmtIf{\bexpr}{\stmtOne}{\stmtTwo}$:
    \begin{align*}
         & \wp{\rwTransP{f}{\stmtIf{\bexpr}{\stmtOne}{\stmtTwo}}}(\expa)                                                                                              \\
         & = \wp{\stmtIf{\bexpr}{\rwTransP{f}{\stmtOne}}{\rwTransP{f}{\stmtTwo}}}(\expa) \tag{definition of $\rwTransP{f}{\cdot}$}                                    \\
         & = \iverson{\bexpr} \cdot \wp{\rwTransP{f}{\stmtOne}}(\expa) + \iverson{\neg \bexpr} \cdot \wp{\rwTransP{f}{\stmtTwo}}(\expa) \tag{definition of $\symWp$}  \\
         & \leq \iverson{\bexpr} \cdot \wp{\rwTransP{g}{\stmtOne}}(\expa) + \iverson{\neg \bexpr} \cdot \wp{\rwTransP{g}{\stmtTwo}}(\expa) \tag{induction hypothesis} \\
         & = \wp{\stmtIf{\bexpr}{\rwTransP{g}{\stmtOne}}{\rwTransP{g}{\stmtTwo}}}(\expa) \tag{definition of $\symWp$}                                                 \\
         & = \wp{\rwTransP{g}{\stmtIf{\bexpr}{\stmtOne}{\stmtTwo}}}(\expa) \tag{definition of $\rwTransP{g}{\cdot}$}
    \end{align*}

    $\bullet$ $\stmt = \stmtWhile{\bexpr}{\stmtP}$:
    Recall that $\wp{\stmtWhile{\bexpr}{\stmtP}}(\expa)$ is defined as the least fixed point of the functional
    \[
        \Phi_\expa(\expb) = \iverson{\bexpr} \cdot \wp{\stmtP}(\expb) + \iverson{\neg \bexpr} \cdot \expa~.
    \]
    By Kleene's fixed point theorem, we have
    \[
        \wp{\rwTransP{f}{\stmtWhile{\bexpr}{\stmtP}}}(\expa) = \sup_{n \in \Nats} \Phi_\expa^n(0)~.
    \]
    Denote by $\Phi_\expa$ the functional for the transformed while loop $\rwTransP{f}{\stmtWhile{\bexpr}{\stmtP}}$ and by $\Psi_\expa$ the functional for $\rwTransP{g}{\stmtWhile{\bexpr}{\stmtP}}$.

    We show by induction on $n$ that $\Phi_\expa^n(0) \leq \Psi_\expa^n(0)$ holds for all $n \in \Nats$.
    The base case $n = 0$ holds since $\Phi_\expa^0(0) = 0 = \Psi_\expa^0(0)$.
    For the induction step, assume $\Phi_\expa^n(0) \leq \Psi_\expa^n(0)$.
    Then,
    \begin{align*}
         & \Phi_\expa^{n+1}(0) = \Phi_\expa(\Phi_\expa^n(0))                                                                                                                      \\
         & = \iverson{\bexpr} \cdot \wp{\rwTransP{f}{\stmtP}}(\Phi_\expa^n(0)) + \iverson{\neg \bexpr} \cdot \expa \tag{definition of $\Phi_\expa$}                               \\
         & \leq \iverson{\bexpr} \cdot \wp{\rwTransP{g}{\stmtP}}(\Psi_\expa^n(0)) + \iverson{\neg \bexpr} \cdot \expa \tag{induction hypothesis and monotonicity of $\wp{\cdot}$} \\
         & = \Psi_\expa(\Psi_\expa^n(0)) \tag{definition of $\Psi_\expa$}                                                                                                         \\
         & = \Psi_\expa^{n+1}(0)~.
    \end{align*}
    Thus, by induction, we have $\Phi_\expa^n(0) \leq \Psi_\expa^n(0)$ for all $n \in \Nats$.
    It follows that
    \begin{align*}
        &\wp{\rwTransP{f}{\stmtWhile{\bexpr}{\stmtP}}}(\expa) = \sup_{n \in \Nats} \Phi_\expa^n(0) \\
        &\leq \sup_{n \in \Nats} \Psi_\expa^n(0) = \wp{\rwTransP{g}{\stmtWhile{\bexpr}{\stmtP}}}(\expa)~.
    \end{align*}
    This concludes the induction.
    Finally, for the full transformation $\rwTrans{f}{\stmt}$, we have
    \begin{align*}
         & \wp{\rwTrans{f}{\stmt}}(\expa)                                                                                               \\
         & = \wp{\stmtAsgn{\rt}{0}\symSemi \stmtTick{f(0)}\symSemi \rwTransP{f}{\stmt}}(\expa) \tag{definition of $\rwTrans{f}{\stmt}$} \\
         & = f(0) + \wp{\rwTransP{f}{\stmt}}(\expa)\substBy{\rt}{0} \tag{definition of $\symWp$}                                        \\
         & \leq g(0) + \wp{\rwTransP{f}{\stmt}}(\expa)\substBy{\rt}{0} \tag{$f(0) \leq g(0)$}                                           \\
         & \leq g(0) + \wp{\rwTransP{g}{\stmt}}(\expa)\substBy{\rt}{0} \tag{\Cref{eq:monotonic-transformations-induction}} \\
         & = \wp{\stmtAsgn{\rt}{0}\symSemi \stmtTick{g(0)}\symSemi \rwTransP{g}{\stmt}}(\expa) \tag{definition of $\symWp$}             \\
         & = \wp{\rwTrans{g}{\stmt}}(\expa) \tag{definition of $\rwTrans{g}{\stmt}$}
    \end{align*}
\end{proof}

\theoremLinearity*

\begin{proof}
    Let $f \colon \PosRealsInf \to \PosRealsInf$ monotonic and $\alpha,\beta \in \PosReals$.
    Denote $g = \alpha \cdot f + \beta$.

    First, we show that auxiliary transformation $\rwTransP{g}{\stmt}$ preserves scalar multiplication.
    For all $\stmt \in \pGCL$ and $\expa \in \Expectations$, we have
    \[
        \wp{\rwTransP{g}{\stmt}}(\alpha \cdot \expa) = \alpha \cdot \wp{\rwTransP{f}{\stmt}}(\expa)~.
    \]
    The proof is by structural induction on $\stmt$.
    The interesting case is $\stmt = \stmtTick{\aexpr}$:
    \begin{align*}
        & \wp{\rwTransP{g}{\stmtTick{\aexpr}}}(\alpha \cdot \expa)                                    \\
        &= \wp{\stmtTick{g(\rt + \aexpr) - g(\rt)}\symSemi \stmtAsgn{\rt}{\rt + \aexpr}}(\alpha \cdot \expa) \tag{def. $\rwTransP{g}{\stmt}$} \\
        &= (g(\rt + \aexpr) - g(\rt)) + \alpha \cdot \expa\substBy{\rt}{\rt + \aexpr}  \tag{def. $\symWp$} \\
        &= ((\alpha \cdot f(\rt + \aexpr) + \beta) - (\alpha \cdot f(\rt) + \beta)) + \alpha \cdot \expa\substBy{\rt}{\rt + \aexpr} \tag{def. $g$} \\
        &= \alpha \cdot (f(\rt + \aexpr) - f(\rt)) + \alpha \cdot \expa\substBy{\rt}{\rt + \aexpr} \tag{cancellation of $\beta$} \\
        &= \alpha \cdot (f(\rt + \aexpr) - f(\rt) + \expa\substBy{\rt}{\rt + \aexpr}) \tag{linearity} \\
        &= \alpha \cdot \wp{\rwTransP{f}{\stmtTick{\aexpr}}}(\expa) \tag{def. of $\symWp$}
     \end{align*}
     The other cases proceed in a straightforward manner.

     For the full transformation $\rwTrans{g}{\stmt}$, we have
     \begin{align*}
        &\wp{\rwTrans{g}{\stmt}}(\alpha \cdot \expa) \\
        &= \wp{\stmtAsgn{\rt}{0}\symSemi \stmtTick{g(0)}\symSemi \rwTransP{g}{\stmt}}(\alpha \cdot \expa) \tag{def. $\rwTrans{g}{\stmt}$} \\
        &= g(0) + \wp{\rwTransP{g}{\stmt}}(\alpha \cdot \expa)\substBy{\rt}{0} \tag{def. $\symWp$} \\
        &= \alpha \cdot f(0) + \beta + \wp{\rwTransP{g}{\stmt}}(\alpha \cdot \expa)\substBy{\rt}{0} \tag{def. $g$} \\
        &= \alpha \cdot f(0) + \alpha \cdot \wp{\rwTransP{f}{\stmt}}(\expa)\substBy{\rt}{0} + \beta \tag{above} \\
        &= \alpha \cdot \wp{\rwTrans{f}{\stmt}}(\expa) + \beta \tag{def. $\rwTransP{f}{\stmt}$}
     \end{align*}
     Using $\expa = 0$, the statement $\wp{\rwTrans{g}{\stmt}}(0) = \alpha \cdot \wp{\rwTrans{f}{\stmt}}(0) + \beta$ follows.
\end{proof}

\theoremGhostBuster*

\begin{proof}
    Let $\stmt \in \pGCL$, $f \colon \PosRealsInf \to \PosRealsInf$ such that $f(x) = \alpha \cdot x + \beta$ for some $\alpha, \beta \in \PosReals$.
    Let $\expa \in \Expectations$ such that $\rt$ does not occur in $\expa$, i.e. $\forall \State \in \States, v \in \Vals.~ \expa\substBy{\rt}{v} = \expa$.

    \begin{align*}
        \intertext{For the transformation by $\rwTransP{\stmt}{f}$ for $\symTick$ statements, we have:}
        & \wp{\rwTransP{f}{\stmtTick{\aexpr}}}(\expa)                                                                                                                        \\
        & = \wp{\stmtTick{(f(\rt + \aexpr) - f(\rt))}\symSemi \stmtAsgn{\rt}{\rt + \aexpr}}(\expa) \tag{definition of $\rwTransP{f}{\stmtTick{\aexpr}}$}                     \\
        & = (f(\rt + \aexpr) - f(\rt)) + \expa\substBy{\rt}{\rt + \aexpr} \tag{definition of $\symWp$}                                                                           \\
        & = (\alpha \cdot (\rt + \aexpr) + \beta - (\alpha \cdot \rt + \beta)) + \expa\substBy{\rt}{\rt + \aexpr} \tag{definition of $f$}         \\
        &= \alpha \cdot \aexpr + \expa\substBy{\rt}{\rt + \aexpr} \tag{cancellation of $\beta$} \\
        &= \alpha \cdot \aexpr + \expa \tag{$\rt$ not in $\expa$} \\
        &= \wp{\stmtTick{\alpha \cdot \aexpr}}(\expa) \tag{definition of $\symWp$}
        \intertext{For the full transformation $\rwTrans{f}{\stmt}$, we have:}
         & \wp{\rwTrans{f}{\stmt}}(\expa)(\State)                                                                                                                             \\
         & = \wp{\stmtAsgn{\rt}{0}\symSemi \stmtTick{f(0)}\symSemi \rwTransP{f}{\stmt}}(\expa)(\State) \tag{definition of $\rwTrans{f}{\stmt}$}                               \\
         & = f(0) + \wp{\rwTransP{f}{\stmt}}(\expa)\substBy{\rt}{0}(\State) \tag{definition of $\symWp$}                                                                      \\
         &= f(0) + \wp{\rwTransP{f}{\stmt}}(\expa) \tag{$\rt$ not in $\wp{\rwTransP{f}{\stmt}}(\expa)$}  \\
         &= (\alpha \cdot 0 + \beta) + \wp{\rwTransP{f}{\stmt}}(\expa) \tag{definition of $f$} \\
         &= \beta + \wp{\rwTransP{f}{\stmt}}(\expa) \tag{simplification} \\
         &= \wp{\stmtTick{\beta}\symSemi \rwTransP{f}{\stmt}}(\expa) \tag{definition of $\symWp$}
    \end{align*}
    Therefore, the simplification is valid and $\rt$ can be omitted for linear transformations.
\end{proof}

\theoremCompositionTransformations*

\begin{proof}
    By definition of $\symWp$, it suffices to show
    \[
        \wp{\rwTransP{g}{\rwTransP{f}{\stmt}}}(\expa)\substBy{\rt'}{f(\rt)} = \wp{\rwTransP{g \circ f}{\stmt}}(\expa\substBy{\rt'}{f(\rt)})~.
    \]
    We proceed by induction over the structure of $\stmt$.

    Of the base cases, $\stmt = \stmtTick{\aexpr}$ is the interesting one:
    \allowdisplaybreaks
    \begin{align*}
         & \wp{\rwTransP{g}{\rwTransP{f}{\stmtTick{\aexpr}}}}(\expa)\substBy{\rt'}{f(\rt)}                                                                                                                                                                                    \\
         & = \wp{\rwTransP{g}{\stmtTick{(f(\rt + \aexpr) - f(\rt))}\symSemi \stmtAsgn{\rt}{\rt + \aexpr}}}(\expa)\substBy{\rt'}{f(\rt)} \tag{definition of $\rwTransP{f}{\cdot}$}                                                                                             \\
         & = \wpleftright{
                \begin{aligned}[c]
                    &\stmtTick{(g(\rt' + (f(\rt + \aexpr) - f(\rt))) - g(\rt'))}\symSemi\\
                    &\stmtAsgn{\rt'}{\rt' + (f(\rt + \aexpr) - f(\rt))}\symSemi\\
                    &\stmtAsgn{\rt}{\rt + \aexpr}
                \end{aligned}
            }(\expa)\substBy{\rt'}{f(\rt)} \tag{definition of $\rwTransP{g}{\cdot}$} \\
         & = (g(\rt' + (f(\rt + \aexpr) - f(\rt))) - g(\rt')                                                                                                                                                                                                                  \\
         & \quad + \expa\substBy{\rt}{\rt + \aexpr}\substBy{\rt'}{\rt' + (f(\rt + \aexpr) - f(\rt))})\substBy{\rt'}{f(\rt)} \tag{definition of $\symWp$}                                                                                                                      \\
         & = g(f(\rt) + (f(\rt + \aexpr) - f(\rt))) - g(f(\rt))                                                                                                                                                                                                               \\
         & \quad + \expa\substBy{\rt}{\rt + \aexpr}\substBy{\rt'}{f(\rt) + (f(\rt + \aexpr) - f(\rt))} \tag{substitution}                                                                                                                                                     \\
         & = g(f(\rt + \aexpr)) - g(f(\rt)) + \expa\substBy{\rt}{\rt + \aexpr}\substBy{\rt'}{f(\rt + \aexpr)} \tag{$f(\rt) + (f(\rt + \aexpr) - f(\rt)) = f(\rt + \aexpr)$}                                                                                                   \\
         & = \wp{\stmtTick{(g(f(\rt + \aexpr)) - g(f(\rt)))}\symSemi \stmtAsgn{\rt'}{f(\rt + \aexpr)}\symSemi \stmtAsgn{\rt}{\rt + \aexpr}}(\expa) \tag{definition of $\symWp$}                                                                                               \\
         & = \wp{\stmtTick{(g(f(\rt + \aexpr)) - g(f(\rt)))}\symSemi \stmtAsgn{\rt}{\rt + \aexpr}\symSemi \stmtAsgn{\rt'}{f(\rt )}}(\expa) \tag{exchange assignments}                                                                                                         \\
         & = \wp{\stmtTick{(g(f(\rt + \aexpr)) - g(f(\rt)))}\symSemi \stmtAsgn{\rt}{\rt + \aexpr}}(\expa\substBy{\rt'}{f(\rt)}) \tag{definition of $\symWp$}                                                                                                                  \\
         & = \wp{\rwTransP{g \circ f}{\stmtTick{\aexpr}}}(\expa\substBy{\rt'}{f(\rt)}) \tag{definition of $\rwTransP{g \circ f}{\cdot}$}
    \end{align*}
    The other base cases are straightforward.

    For the inductive cases, assume that the claim holds for $\stmtOne$ and $\stmtTwo$.

    We show the case for sequential composition $\stmt = \stmtSeq{\stmtOne}{\stmtTwo}$:
    \begin{align*}
         & \wp{\rwTransP{g}{\rwTransP{f}{\stmtSeq{\stmtOne}{\stmtTwo}}}}(\expa)\substBy{\rt'}{f(\rt)}                                                                      \\
         & = \wp{\rwTransP{g}{\rwTransP{f}{\stmtOne}\symSemi \rwTransP{f}{\stmtTwo}}}(\expa)\substBy{\rt'}{f(\rt)} \tag{definition of $\rwTransP{f}{\cdot}$}               \\
         & = \wp{\rwTransP{g}{\rwTransP{f}{\stmtOne}}\symSemi \rwTransP{g}{\rwTransP{f}{\stmtTwo}}}(\expa)\substBy{\rt'}{f(\rt)} \tag{definition of $\rwTransP{g}{\cdot}$} \\
         & = \wp{\rwTransP{g}{\rwTransP{f}{\stmtOne}}}(\wp{\rwTransP{g}{\rwTransP{f}{\stmtTwo}}}(\expa))\substBy{\rt'}{f(\rt)} \tag{definition of $\symWp$}                \\
         & = \wp{\rwTransP{g \circ f}{\stmt}}(\wp{\rwTransP{g}{\rwTransP{f}{\stmtTwo}}}(\expa)\substBy{\rt'}{f(\rt)}) \tag{induction hypothesis for $\stmtOne$}            \\
         & = \wp{\rwTransP{g \circ f}{\stmt}}(\wp{\rwTransP{g \circ f}{\stmtTwo}}(\expa\substBy{\rt'}{f(\rt)})) \tag{induction hypothesis for $\stmtTwo$}                  \\
         & = \wp{\rwTransP{g \circ f}{\stmtOne}\symSemi \rwTransP{g \circ f}{\stmtTwo}}(\expa\substBy{\rt'}{f(\rt)}) \tag{definition of $\symWp$}                          \\
         & = \wp{\rwTransP{g \circ f}{\stmt}}(\expa\substBy{\rt'}{f(\rt)}) \tag{definition of $\rwTransP{g \circ f}{\cdot}$}
    \end{align*}
    The other inductive cases are straightforward.
    Hence, the claim follows by induction.
\end{proof}

\thmMcTransformationSound*

\begin{proof}
    We one-to-one relate every finite path $\pi = \State_1 \ldots \State_n$ of the original Markov chain $\mathcal{M}$ to a finite path $\pi' = \State_1' \ldots \State_n'$ in the transformed Markov chain $\rwTrans{f}{\mathcal{M}}$.
    We use two properties of \Cref{lemma:mc-transformation-sound} below to show that the expected rewards of the two paths are equal.
    \begin{enumerate}
        \item[(1)] $\pi$ and $\pi'$ have the same probabilities: $\Prob(\pi) = \ProbP(\pi')$.
        \item[(3)] The $f$-transformed reward is equal to the reward of $\pi'$: $f(\Rew{\pi}) = \RewP{\pi'}$.
    \end{enumerate}
    We obtain:
    \begin{align*}
        \phantom{=} \Expected(f(\Rew{\mathcal{M}}))  & \eeq \sup_{n \in \Nats} \sum_{\pi \in \PathsN{\mathcal{M}}{=n}(\State)} \Prob(\pi) \cdot f(\Rew{\pi})                                \\
        = \Expected(\RewP{\rwTrans{f}{\mathcal{M}}}) & \eeq \sup_{n \in \Nats} \sum_{\pi' \in \PathsN{\rwTrans{f}{\mathcal{M}}}{=n}(\State)} \ProbP(\pi') \cdot \RewP{\pi'}~. &  & \qedhere
    \end{align*}
\end{proof}

\lemmaMcTransformationSound*

\begin{proof}
    Let $\mathcal{M} = (\States, \Prob, \State_1, \symRew)$ be a Markov chain and let $f \colon \PosRealsInf \to \PosRealsInf$ be a monotonic function with $f(0) = 0$.
    \newcommand{\gMap}[1]{\xi_{#1}}

    For each $n \geq 1$, define
    \[
        \begin{aligned}
            \gMap{n} &\colon \PathsN{\mathcal{M}}{=n} \to \PathsN{\rwTrans{f}{\mathcal{M}}}{=n}, \\
            \gMap{n}(\State_1 \ldots \State_n) &\coloneqq (\State_1,\alpha_1) \ldots (\State_n,\alpha_n),
        \end{aligned}
    \]
    where $\alpha_1 = 0$ and $\alpha_i = \sum_{j < i}\Rew{\State_j}$ for all $1 \leq i \leq n$.
    For every path $\pi = \State_1 \ldots \State_n \in \PathsN{\mathcal{M}}{=n}$, we prove by induction on $n$ that $\gMap{n}(\pi)$ is well-defined and satisfies properties (1)--(3).
    \medskip

    \noindent
    \textit{Base case.}
    For $n = 1$, let $\State_1 \in \PathsN{\mathcal{M}}{=1}$.
    Then, $\gMap{1}(\State_1) = (\State_1, 0) \in \PathsN{\rwTrans{f}{\mathcal{M}}}{=1}$ and the properties (1)--(3) hold:
    \begin{enumerate}
        \item $\ProbP((\State_1,0)) = 1 = \Prob(\State_1)$,
        \item $\alpha_1 = 0 = \sum_{j < 1} \Rew{\State_j}$,
        \item $\RewP{(\State_1,0)} = f(0 + \Rew{\State_1}) - f(0) = f(\Rew{\State_1})$,
    \end{enumerate}
    \medskip

    \noindent
    \textit{Induction step.}
    Now assume the claim for length $n$ and let $\pi = \State_1 \ldots \State_{n+1} \in \PathsN{\mathcal{M}}{=n+1}$.
    We define the truncated path $\pi$ up to $n$ steps by $\pi^- = \State_1 \ldots \State_n$.
    Then, $\gMap{n}(\pi^-) = (\State_1,\alpha_1) \ldots (\State_n,\alpha_n)$.
    With $\alpha_{n+1} \coloneqq \alpha_n + \Rew{\State_n}$, we have
    \[
        \gMap{n+1}(\pi) = (\State_1,\alpha_1) \ldots (\State_n,\alpha_n)(\State_{n+1},\alpha_{n+1}).
    \]
    The path $\gMap{n+1}(\pi)$ is a valid path in $\rwTrans{f}{\mathcal{M}}$ since by definition of $\ProbP$, the only successor state of $(\State_n,\alpha_n)$ with nonzero probability is exactly $(\State_{n+1},\alpha_{n+1})$.
    Furthermore, we show that $\gMap{n+1}(\pi)$ satisfies properties (1)--(3):
    \begin{enumerate}
        \item Same path probabilities:
              \begin{align*}
                  \ProbP(\gMap{n+1}(\pi))
                   & = \ProbP(\gMap{n}(\pi^-)) \cdot \ProbP((\State_n,\alpha_n), (\State_{n+1},\alpha_{n+1})) \\
                   & = \Prob(\pi^-) \cdot \ProbP((\State_n,\alpha_n), (\State_{n+1},\alpha_{n+1})) \tag{IH (1)} \\
                   & = \Prob(\pi^-) \cdot \Prob(\State_n)(\State_{n+1}) \tag{def.\ $\ProbP$}                    \\
                   & = \Prob(\pi)~.
              \end{align*}
        \item $\alpha_i$ track untransformed rewards:
              By IH (2), we have $\alpha_n  = \sum_{j < n} \Rew{\State_j}$.
              Thus, we have $\alpha_{n+1} = \left( \sum_{j < n} \Rew{\State_j} \right) + \Rew{\State_n} = \sum_{j < n+1} \Rew{\State_j}$.
        \item Transformed rewards:
              \begin{align*}
                   & \RewP{\gMap{n+1}(\pi)} \\
                   & = \RewP{\gMap{n}(\pi^-)\,(\State_{n+1},\alpha_{n+1})} \tag{def.\ $\gMap{n+1}$}                                        \\
                   & = \RewP{\gMap{n}(\pi^-)} + \symRewP((\State_{n+1},\alpha_{n+1})) \tag{def.\ $\symRewP$}                         \\
                   & = f(\Rew{\pi^-}) + \symRewP((\State_{n+1},\alpha_{n+1})) \tag{IH (3)}                                                   \\
                   & = f(\Rew{\pi^-}) + \left( f(\alpha_{n+1} + \Rew{\State_{n+1}}) - f(\alpha_{n+1}) \right) \tag{def.\ $\symRewP$}       \\
                   & = f(\Rew{\pi^-}) + \left( f((\alpha_n + \Rew{\State_n}) + \Rew{\State_{n+1}}) - f(\alpha_n + \Rew{\State_n}) \right) \tag{def.\ $\alpha_{n+1}$} \\
                   & = f(\Rew{\pi^-}) + \left( f(\Rew{\pi^-} + \Rew{\State_{n+1}}) - f(\Rew{\pi^-}) \right) \tag{IH (2)}                  \\
                   & = f(\Rew{\pi^-}) + \left( f(\Rew{\pi}) - f(\Rew{\pi^-}) \right) \tag{def.\ $\Rew{\pi}$}                               \\
                   & = f(\Rew{\pi})~.
              \end{align*}
    \end{enumerate}
    Thus, $\gMap{n+1}$ is well-defined and satisfies (1)--(3).
    Moreover, $\gMap{n+1}$ is a bijection.
    To show injectivity, let $\pi,\pi' \in \PathsN{\mathcal{M}}{=n+1}$ with $\gMap{n+1}(\pi)=\gMap{n+1}(\pi')$.
    By definition of $\gMap{n+1}$, it follows that $\pi = \pi'$.
    For surjectivity, let $\rho = (\State_1,\beta_1)\ldots(\State_{n+1},\beta_{n+1}) \in \PathsN{\rwTrans{f}{\mathcal{M}}}{=n+1}$ be arbitrary.
    It holds that $\gMap{n+1}(\State_1\ldots\State_{n+1}) = \rho$.
\end{proof}
To lift the restriction $f(0) = 0$ in \Cref{thm:mc-transformation-sound}, we use the following lemma.
It uses a shifted function $g(x) = f(x) - f(0)$ for the transformation, and prepends an initial state with reward $f(0)$ to the transformed Markov chain.

\begin{lemma}[Lifting the restriction $f(0) = 0$]\label{lem:mc-transformation-sound-lifted}
    Consider a Markov chain $\mathcal{M} = (\States, \Prob, \State_1, \symRew)$ and let $f \colon \PosRealsInf \to \PosRealsInf$ be a monotonic function.
    Let $g \colon \PosRealsInf \to \PosRealsInf$ be defined as $g(x) = f(x) - f(0)$.
    From \Cref{thm:mc-transformation-sound}, we obtain the $g$-transformed Markov chain $\rwTrans{g}{\mathcal{M}} = (\States \times \PosRealsInf,~ \ProbP,~ \State_1',~ \symRewP )$.
    To insert an initial state with reward $f(0)$, let $\mathcal{M}'' = (\States'',~ \ProbP',~ \State_0'',~ \symRewP')$ with:
    \begin{itemize}
        \item $\States \times \PosRealsInf \cup \Set{\State_0''}$,
        \item $\ProbP'(\State, \State') = \begin{cases}
                      1                       & \text{if } \State = \State_0' \land \State' = (\State_1, 0) \\
                      \ProbP(\State, \State') & \text{otherwise}
                  \end{cases}$
        \item $\symRewP'(\State_0'') = f(0)$ and otherwise $\symRewP'((\State, \alpha)) = \symRewP((\State, \alpha))$.
    \end{itemize}
    Then, it holds that
    \[
        \Expected(\Rew{\mathcal{M}''}) = \Expected(f(\Rew{\mathcal{M}}))~.
    \]
\end{lemma}

\begin{proof}
    Since $g$ is monotonic and $g(0) = f(0) - f(0) = 0$, we can apply \Cref{thm:mc-transformation-sound} to obtain the $g$-transformed Markov chain $\rwTrans{g}{\mathcal{M}}$.
    In the extended MC, every path's reward sequence is prepended with $f(0)$.
    Therefore,
    \begin{align*}
        \Expected(\Rew{\mathcal{M}''})
         & = \Expected(f(0) + \RewP{\rwTrans{g}{\mathcal{M}}})                                  \\
         & = \Expected(f(0) + g(\Rew{\mathcal{M}})) \tag{by \Cref{thm:mc-transformation-sound}} \\
         & = \Expected(f(0) + f(\Rew{\mathcal{M}}) - f(0)) \tag{definition of $g$}              \\
         & = \Expected(f(\Rew{\mathcal{M}})) \tag*{\qedhere}
    \end{align*}
\end{proof}

\subsection[Generalization to Multiple Rewards (Section~\ref{sec:multiple-rewards})]{Generalization to Multiple Rewards (\Cref{sec:multiple-rewards})}\label{sec:proofs-multiple-rewards}

\begingroup
\newcommand{\symWpN}{\symWp^{n}}
\newcommand{\wpN}[1]{\symWpN\llbracket #1\rrbracket}
\newcommand{\rwTransN}[1]{\rwTrans{f}{#1}}
\newcommand{\rwTransPN}[1]{\rwTransP{f}{#1}}
\newcommand{\rtVec}{\multi{\rt}}
\newcommand{\rtVecP}{\multi{\rt}'}
\newcommand{\rtVecZero}{(0,\ldots,0)}

We outline the generalization of our results to multiple rewards, corresponding to \Cref{sec:multiple-rewards}.
Fix $n \geq 1$ and monotone $f \colon (\PosRealsInf)^n \to \PosRealsInf$.

\paragraph{Multi-reward probabilistic programs and Markov chains.}
In the operational and weakest pre-expectation semantics (\Cref{sec:foundations}), reward statements become tuple-valued: $\stmtTick{(\aexpr_1,\ldots,\aexpr_n)}$.
Markov chains therefore carry tuple-valued rewards:
\[
    \mathcal{M} = (\States, \Prob, \State_1, \symRew),
    \qquad
    \symRew \colon \States \to (\PosRealsInf)^n.
\]
A finite path $\pi = \State_1 \ldots \State_k$ has reward $\Rew{\pi} \coloneqq \sum_{i=1}^{k} \symRew(\State_i) \in (\PosRealsInf)^n$.
For each $i$, let $\stmt_i$ be obtained from $\stmt$ by replacing each tuple reward by its $i$-th component.
Then, the tuple-valued weakest pre-expectation transformer is defined componentwise by
\begin{align*}
    \symWpN &\colon \pGCL \times \Expectations^n \to \Expectations^n, \\
    \wpN{\stmt}(\expa_1,\ldots,\expa_n) &\coloneqq \left(\wp{\stmt_1}(\expa_1),\ldots,\wp{\stmt_n}(\expa_n)\right).
\end{align*}

\paragraph{Generalizing the transformation.}
Recall that the reward program transformation (\Cref{def:reward-program-transformation}) extends as explained in \Cref{sec:multiple-rewards}: we introduce fresh ghost variables $\rt_1,\ldots,\rt_n$ and define
\[
    \rwTransN{\stmt} \coloneqq \stmtAsgn{\rt_1}{0}\symSemi \cdots \symSemi \stmtAsgn{\rt_n}{0}\symSemi \stmtTick{f(0,\ldots,0)}\symSemi \rwTransPN{\stmt}~,
\]
where for $\symTick$ statements, the inner transformation is defined as follows:
\begin{align*}
    \rwTransPN{\stmtTick{(\aexpr_1,\ldots,\aexpr_n)}} \coloneqq~& \stmtTick{f(\rt_1+\aexpr_1,\ldots,\rt_n+\aexpr_n)-f(\rt_1,\ldots,\rt_n)}\symSemi \\
    &\stmtAsgn{\rt_1}{\rt_1+\aexpr_1}\symSemi \cdots \symSemi \stmtAsgn{\rt_n}{\rt_n+\aexpr_n}~.
\end{align*}
The other cases are defined analogously to the single-reward case (c.f. \Cref{fig:reward-transformation}), by applying the transformation recursively.
Similarly, the MC transformation (\Cref{thm:defi-mc-transformation}) extends as follows, assuming $f(0,\ldots,0) = 0$:
\begin{align*}
    \rwTransN{\mathcal{M}} ={}& (\States \times (\PosRealsInf)^n,~ \ProbP,~ (\State_1,\rtVecZero),~ \symRewP), \\
    \ProbP((\State,\rtVec),(\State',\rtVecP)) ={}& \begin{cases}
                                                       \Prob(\State)(\State') & \text{if }\rtVecP = \rtVec + \symRew(\State),\\
                                                       0                      & \text{otherwise,}
                                                   \end{cases} \\
    \symRewP((\State,\rtVec)) ={}& f(\rtVec + \symRew(\State)) - f(\rtVec)~,
\end{align*}
Soundness of the Markov chain transformation (\Cref{lemma:mc-transformation-sound,lem:mc-transformation-sound-lifted}) lifts directly to tuple-valued accumulated rewards, using the same path coupling and telescoping arguments as in the scalar case, and we obtain
\[
    \Expected(f(\Rew{\mathcal{M}})) = \Expected(\RewP{\rwTransN{\mathcal{M}}}).
\]
\Cref{thm:correctness-reward-translation}, i.e. $\Expected(f(\Rew{\stmt})) = \Expected(\Rew{\rwTransN{\stmt}})$, follows by the same argument.
\endgroup

%% file: appendices/case_studies.tex
\section{Case Studies}\label{sec:case-studies}

We give additional details on the examples mentioned in the main text.
For the examples verified with the \emph{Caesar} verifier~\cite{DBLP:journals/pacmpl/SchroerBKKM23}, we provide the associated code in its \emph{HeyVL} language.
The examples run on an unmodified version of Caesar (version 3.0.0), making use of its support for $\symWp$-based reasoning about expected rewards with $\symTick$ statements.

\paragraph{Syntax of HeyVL.}
We give a short explanation of the syntax of HeyVL that is used in the examples below.
More detailed information can be found in the associated publication \cite{DBLP:journals/pacmpl/SchroerBKKM23} and the online documentation (\url{https://caesarverifier.org/docs/heyvl/}).

In the examples below, we use Caesar's syntax to declare two kinds of procedures in HeyVL, $\symProc$ and $\symcoProc$.
Both kinds of declarations consist of a name, a list of input and output variables, and a pre $\expa$, a post $\expb$, and a body $\stmt$.
The syntax of a $\symProc$ is as follows, where \lstinline!x! is an input variable of type \lstinline!Int! and \lstinline!y! is an output variable of type \lstinline!Int!:
\begin{lstlisting}[language=HeyVL,mathescape]
proc exampleProc(x: Int) -> (y: Int)
    pre $\expa$
    post $\expb$
{
    $\stmt$
}
\end{lstlisting}
The only difference between $\symProc$s and $\symcoProc$s is that $\symProc$s are verified to satisfy \emph{lower bounds} $\expa$ on expected values of $\expb$, i.e.
\[
    \expa \expleq \wp{\stmt}(\expb)~,
\]
and that $\symcoProc$ s are required to satisfy the dual \emph{upper bound} property, i.e.
\[
    \wp{\stmt}(\expb) \expleq \expa~.
\]

To implement probabilistic choices, HeyVL uses the syntax of an ordinary assignment with a \emph{distribution expression} $\exprFlip{p}$ that returns $\exprTrue$ with probability $p$ and $\exprFalse$ with probability $1 - p$.

\clearpage
\subsection[Expected Visiting Times for the Fast Dice Roller (Section~\ref{sec:expected-visiting-times})]{Expected Visiting Times for the Fast Dice Roller (\Cref{sec:expected-visiting-times})}\label{sec:fdr-evt}

To illustrate the encoding for expected visiting times from \Cref{sec:expected-visiting-times}, we consider the \emph{fast dice roller} (FDR), a runtime-optimal algorithm for sampling uniformly from $[1,N]$~\cite{Lum13}.

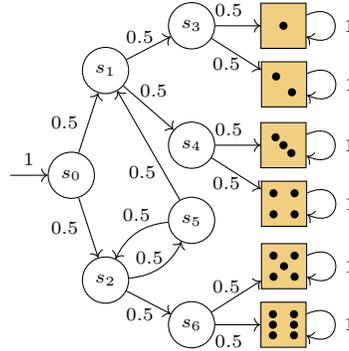
\begin{wrapfigure}{r}{0.41\textwidth}
    \vspace{-\baselineskip}
    \centering
    \input{appendices/fdr_markov_chain.tex}
    \caption{Markov chain of the fast dice roller (FDR) from \cite{DBLP:journals/jar/MertensKQW25}.}
    \label{fig:fdr-markov-chain}
    \vspace{-1.1\baselineskip}
\end{wrapfigure}
For $N=6$, \Citeauthor{DBLP:journals/jar/MertensKQW25}~\cite{DBLP:journals/jar/MertensKQW25} analyze expected visiting times on the Markov chain shown in \Cref{fig:fdr-markov-chain}.
For this example, they compute expected visiting times to reach the states $s_3,s_4,s_6$ that terminate the algorithm, and from these deduce the expected reachability probabilities of the dice faces.
We verify the same bounds on expected visiting times using Caesar, by encoding the Markov chain as simple $\symWhile$ loop with probabilistic choices and $\symTick$ statements.

\Cref{fig:fdr-evt-heyvl} shows the HeyVL code for this encoding.
Variable \lstinline|s| models the state index, \lstinline|done| indicates whether a terminal state has been reached, and \lstinline|res| models the sampled dice face.
To query the expected visiting time of specific states, we use ghost variables \lstinline!query_s! and \lstinline!query_done!, and collect reward with $\stmtTick{\iverson{\texttt{query\_s} = s \land \texttt{query\_done} = \texttt{done}}}$.
The $\symcoProc$ declaration in \Cref{fig:fdr-evt-heyvl} with $\Requires{1/3}$ specifies that we verify the upper bound of $\nicefrac{1}{3}$ on the expected visiting time to reach states $s \in \{0,\ldots,6\}$ and $done \in \{\exprTrue, \exprFalse\}$.

From these bounds, we obtain bounds on the expected reachability probabilities of the terminal states (c.f. \cite[Section~6.1.1]{DBLP:journals/jar/MertensKQW25}).
Since the relevant states $s_3,s_4,s_6$ have expected visiting time at most $\nicefrac{1}{3}$, and each transition from these states to a terminal state has probability $\nicefrac{1}{2}$, the expected reachability probability of each terminal state is at most $\nicefrac{1}{6}$.
One can also verify this last fact by adding more ghost variables to indicate the specific value of the dice face (\lstinline!res! variable) and modifying the $\symTick$ statements accordingly.

The invariant is omitted in \Cref{fig:fdr-evt-heyvl} for brevity, but given separately in \Cref{fig:fdr-evt-invariant}.
It specifies, for each queried state and each current state, an upper bound on the expected visiting time from the current state to the queried state.
It corresponds to the \emph{fundamental matrix}, defined in \cite[Def. 3.2.2]{Kemeny1976-fa}, originally given as a definition to determine expected visiting times in Markov chains.

We also confirmed with Caesar that our invariant is \emph{optimal}.
This was done using the Caesar built-in dual to the $\symWp$ Park induction proof rule for \emph{upper} bounds; now we verify a \emph{lower} bound.
Since both lower and upper bounds coincide, the invariant is optimal.

\clearpage
\subsection[Variance of Web Server (Example~\ref{ex:variance-webserver})]{Variance of Web Server (\Cref{ex:variance-webserver})}\label{sec:webserver-variance}

For the web server examples in \Cref{fig:example-programs}, we made claims about the expected value and variance of the response time.
Some details on the calculations were provided in \Cref{ex:variance-webserver}, and we add some selected calculations here.

We focus on the unmodified example, i.e. \Cref{fig:example-program-a}.
For the $\symWhile$ loop, the loop-characteristic functional w.r.t. post $0$ is given by
\[
    \Phi_0(\expa) \eeq \iverson{\neg done} \cdot ((2 \cdot \rt + 1) + 0.5 \cdot \expa\substBy{done}{\exprTrue} + 0.5 \cdot \expa) + \iverson{done} \cdot 0.
\]
To compute the least fixed-point, we make use of Kleene's fixed point theorem:
\[
    \wp{\stmtWhile{\neg done}{\ldots}}(0) \eeq \sup_{n \in \Nats} \Phi_0^n(0).
\]
The first few iterations give the following results:
\begin{align*}
    \Phi_0^0(0)                    & \eeq 0,                                                                                                                          \\
    \Phi_0^1(0)                    & \eeq \iverson{\neg done} \cdot (2 \cdot \rt + 1),                                                                                \\
    \Phi_0^2(0)                    & \eeq \iverson{\neg done} \cdot (\nicefrac{1}{2} \cdot (6 \cdot \tau + 5)),                                                       \\
    \Phi_0^3(0)                    & \eeq \iverson{\neg done} \cdot (\nicefrac{1}{4} \cdot (14 \cdot \tau + 15)),                                                     \\
    \Phi_0^4(0)                    & \eeq \iverson{\neg done} \cdot (\nicefrac{1}{8} \cdot (30 \cdot \tau + 37))~.
    \intertext{From this, we can deduce the following closed form:}
    \Phi_0^n(0)                    & \eeq \iverson{\neg done} \cdot \left(\frac{1}{2^{n-1}} \cdot ((2^{n+1} - 2) \cdot \tau + 3 \cdot (2^n - 1) - 2 \cdot n)\right)~.
    \intertext{We make use of the monotone convergence theorem to compute the supremum:}
    \sup_{n \in \Nats} \Phi_0^n(0) & \eeq \lim_{n \to \infty} \Phi_0^n(0)                                                                                             \\
                                   & \eeq \iverson{\neg done} \cdot (4 \cdot \tau + 6)~.
\end{align*}
Hence, we have shown that $\wp{\stmtWhile{\neg done}{\ldots}}(0) \eeq \iverson{\neg done} \cdot (4 \cdot \tau + 6)$.
The full program $\stmt$ starts with $\rt = 0$ and $done = \exprFalse$, therefore $\wp{\rwTrans{x^2}{\stmt}}(0) = 6$.
For a random variable $X$, we have $\mathsf{Var}(X) = \mathsf{E}(X^2) - \mathsf{E}(X)^2$.
Hence, we compute $\wp{\rwTrans{x^2}{\stmt}}(0) - \wp{\stmt}(0)^2 = 6 - 2^2 = 2$.

For the example program with the cache (\Cref{fig:example-program-b}), denoted by $\stmtP$, we obtain the least fixed-point  $\iverson{\neg done} \cdot (3 \cdot (\rt - 0.5) + 4.5)$ and obtain the second moment $\wp{\rwTrans{x^2}{\stmtP}}(0) = 0.25 + 4.5 = 4.75$.
Since $\wp{\stmtP}(0) = 2$, we obtain the variance $4.75 - 2^2 = 0.75$.

The associated HeyVL code is shown in \Cref{fig:example-programs-heyvl}.
\Cref{fig:example-programs-heyvl-a,fig:example-programs-heyvl-b} verify the expected runtime upper bound of $2$ for the unmodified and modified examples, respectively.
\Cref{fig:example-programs-heyvl-c,fig:example-programs-heyvl-d} verify the second moment upper bounds of $6$ and $4.75$, respectively.

\subsection[Biased Random Walk (Example~\ref{ex:biased-random-walk-second-moment})]{Biased Random Walk (\Cref{ex:biased-random-walk-second-moment})}

Recall \Cref{ex:biased-random-walk-second-moment} in \Cref{sec:higher-moments-of-rewards}, where we consider a biased random walk on the non-negative integers, starting at some initial value $\texttt{init\_x}$.
With probability $0.75$, the variable \lstinline!x! is decremented by $2$, and with probability $0.25$ it is incremented by $2$.
After each step, we increment the reward by $1$.
It stops when \lstinline!x! reaches zero.

The HeyVL code for the first moment of the cumulative reward is shown in \Cref{fig:biased-random-walk-second-moment}.
The upper bound shown is \lstinline!init_x - (init_x \% 2)!, which means the expected number of steps is upper-bounded by the initial value of \lstinline!x!, rounded down to the nearest even number.

For the second moment, we apply the transformation $\rwTransP{f}{\cdot}$ to the original program, where $f(x) = x^2$.
The resulting program is shown in \Cref{fig:biased-random-walk-second-moment}.
As before, the original $\symTick$ statement is replaced by $\stmtTick{(\texttt{tau} + 1)^2 - \texttt{tau}^2}$ and the assignment $\stmtAsgn{\texttt{tau}}{\texttt{tau} + 1}$.
We used the invariant $x^2 + 2 \cdot x \cdot \texttt{tau} + 3 \cdot x$ to to verify the upper bound of $x^2 + 3 \cdot x$.

\subsection[Cumulative Distribution Function (Example~\ref{ex:cdf-webserver})]{Cumulative Distribution Function (\Cref{ex:cdf-webserver})}

For the web server example in \Cref{fig:example-program-a}, \Cref{ex:cdf-webserver} in \Cref{sec:cdf} verifies a bound on the runtime CDF.
We encode this objective by applying the transformation with $f(x)=\iverson{x \geq N}$.
In the transformed program (\Cref{fig:cdf-webserver-heyvl}), $\tau$ tracks runtime and each reward increment is
\[
    \iverson{\tau+1\geq N}-\iverson{\tau\geq N} \quad=\quad \iverson{\tau+1=N}~,
\]
so reward is collected exactly when the run first reaches step $N$.
We use the following loop invariant to verify that $\Pr(\text{runtime} \geq N)$ is at most $0.5^{N \monus 1}$:
\[
    I \eeq \iverson{\neg done \land \tau + 1 \leq N} \cdot \left(\frac{1}{2}\right)^{N - (\tau + 1)}~.
\]
Note that this invariant requires exponentials, which are not supported by Caesar out of the box, but can be encoded by a user-defined function.
This is done using $\symStmt{\texttt{domain}}~\texttt{Exponentials}$ declaration, which contains a recursive definition of \texttt{pow}, modeling the function $a^b$ for $a \in \PosReals$ and $b \in \Nats$.

\subsection[Expected Excess (Example~\ref{ex:loop-splitting-excess})]{Expected Excess (\Cref{ex:loop-splitting-excess})}

Recall \Cref{ex:loop-splitting-excess} in \Cref{sec:expected-excess}: for the web server loop from \Cref{fig:example-program-a} with per-round success probability $p$ and budget $N$, the objective is $\Expected(\text{runtime} \monus N)$.

\Cref{fig:expected-excess-original} encodes the original program, \Cref{fig:expected-excess-single} encodes the transformed single-loop variant with $f(x)=x \monus N$ (hence increment $\iverson{\tau \geq N}$), and \Cref{fig:expected-excess-split} encodes the split two-phase variant.
The two transformed encodings both prove the same upper bound $\frac{(1-p)^N}{p}$ for all $0 < p < 1$.
Again, the invariants for the transformed programs require exponentials, which are encoded by a user-defined function as described in the previous example.

\subsection[MGF for Coin Flip (Example~\ref{ex:mgf-coin-flip})]{MGF for Coin Flip (\Cref{ex:mgf-coin-flip})}

Recall \Cref{ex:mgf-coin-flip} in \Cref{sec:mgf}: for $\stmt = \stmtProb{p}{\stmtAsgn{x}{1}}{\stmtAsgn{x}{0}}\symSemi \stmtTick{x}$ and $f(x)=e^{t \cdot x}$, the objective is to bound the MGF by $p \cdot e^t + (1-p)$.

\Cref{fig:mgf-coin-flip-heyvl} gives the corresponding Caesar proof encoding in HeyVL.
Its $\symStmt{\texttt{domain}}~\texttt{MGFs}$ declaration introduces an uninterpreted \lstinline!mgf! for $e^{t \cdot x}$ with axioms \lstinline!mgf(0)=1! and strict monotonicity (\lstinline!x < y ==> mgf(x) < mgf(y)!).
These axioms fix the initial reward and rule out countermodels such as \lstinline!mgf(0)=1, mgf(1)=0!, where \lstinline!mgf(tau + x) - mgf(tau)! is interpreted in \lstinline!UReal! as \lstinline!max(0, mgf(tau + x)-mgf(tau))! and the increment is truncated to \lstinline!0!.

\subsection[Multiple Rewards (Example~\ref{ex:multiple-rewards-cost})]{Multiple Rewards (\Cref{ex:multiple-rewards-cost})}

Recall \Cref{ex:multiple-rewards-cost} in \Cref{sec:multiple-rewards}: we track runtime and memory, starting at $2$ each, then add runtime with probability $p$ and additional memory with probability $q$.
\Cref{fig:multiple-rewards-multiplication} gives the corresponding Caesar proof encoding in HeyVL.
Using $f(memory, runtime)=memory \cdot runtime$, each update contributes the incremental product difference $f(memory',runtime')-f(memory,runtime)$.
The file verifies the upper bound $4 + (p \cdot (2 + (q \cdot 3)))$ for the expected combined cost.

\begin{figure}[p]
\begin{lstlisting}[language=HeyVL]
coproc fdr6(query_s: UInt, query_done: Bool) -> (res: UInt)
    pre !?(query_done == true && query_s <= 6)
    pre 1/3
    post 0
{
    var s: UInt = 0
    var done: Bool = false
    @invariant(...)
    while !done {
        reward [query_s == s && query_done == done]
        var choice: Bool; var old_s: UInt = s
        if old_s == 0 {
            choice = flip(0.5)
            if choice { s = 1 } else { s = 2 }
        } else {}; if old_s == 1 {
            choice = flip(0.5)
            if choice { s = 3 } else { s = 4 }
        } else {}; if old_s == 2 {
            choice = flip(0.5)
            if choice { s = 5 } else { s = 6 }
        } else {}; if old_s == 3 {
            choice = flip(0.5)
            if choice { res = 1 } else { res = 2 }
            done = true
        } else {}; if old_s == 4 {
            choice = flip(0.5)
            if choice { res = 3 } else { res = 4 }
            done = true
        } else {}; if old_s == 5 {
            choice = flip(0.5)
            if choice { s = 1 } else { s = 2 }
        } else {}; if old_s == 6 {
            choice = flip(0.5)
            if choice { res = 5 } else { res = 6 }
            done = true
        } else {}
    }
    reward [query_s == s && query_done == done]
}
\end{lstlisting}
    \caption{HeyVL code for verifying expected visiting times in the FDR example. The invariant is given in \Cref{fig:fdr-evt-invariant}.}
    \label{fig:fdr-evt-heyvl}
\end{figure}

\begin{figure}[p]
    {\tiny{}
    \begin{lstlisting}[language=HeyVL,gobble=8]
        ite(
            done,
            /* then */ [((query_s == s) && (query_done == true))],
            /* else */
            ite(
            (query_done == false),
            /* then */
                [query_s==0] * [s==0]
                + [query_s==1] * ( [s==0]*(2/3) + [s==1]*(1)   + [s==2]*(1/3) + [s==5]*(2/3) )
                + [query_s==2] * ( [s==0]*(2/3) + [s==2]*(4/3) + [s==5]*(2/3) )
                + [query_s==3] * ( [s==0]*(1/3) + [s==1]*(1/2) + [s==2]*(1/6) + [s==3]*1 + [s==5]*(1/3) )
                + [query_s==4] * ( [s==0]*(1/3) + [s==1]*(1/2) + [s==2]*(1/6) + [s==4]*1 + [s==5]*(1/3) )
                + [query_s==5] * ( [s==0]*(1/3) + [s==2]*(2/3) + [s==5]*(4/3) )
                + [query_s==6] * ( [s==0]*(1/3) + [s==2]*(2/3) + [s==5]*(1/3) + [s==6]*1 ),
            /* else */
                [query_s==3] * ( [s==0]*(1/3) + [s==1]*(1/2) + [s==2]*(1/6) + [s==3]*1 + [s==5]*(1/3) )
                + [query_s==4] * ( [s==0]*(1/3) + [s==1]*(1/2) + [s==2]*(1/6) + [s==4]*1 + [s==5]*(1/3) )
                + [query_s==6] * ( [s==0]*(1/3) + [s==2]*(2/3) + [s==5]*(1/3) + [s==6]*1 )
            )
        )
    \end{lstlisting}}

    \caption{Invariant for the FDR expected visiting time verification in \Cref{fig:fdr-evt-heyvl}.}
    \label{fig:fdr-evt-invariant}
\end{figure}

\begin{figure}[p]
    \centering
    \begin{subfigure}[b]{0.45\textwidth}
\begin{lstlisting}[language=HeyVL,linerange={1-11}]
coproc original_moment1() -> (done: Bool)
    pre 2
    post 0
{
    done = false
    @invariant([!done] * 2)
    while !done {
        reward 1
        done = flip(0.5)
    }
}
\end{lstlisting}
        \caption{HeyVL code to verify that the expected runtime of the unmodified web server example (\Cref{fig:example-program-a}) is upper-bounded by $2$.}
        \label{fig:example-programs-heyvl-a}
    \end{subfigure}
    \hfill
    \begin{subfigure}[b]{0.45\textwidth}
\begin{lstlisting}[language=HeyVL,linerange={13-24}]
coproc caching_moment1() -> (done: Bool)
    pre 2
    post 0
{
    done = false
    reward 0.5
    @invariant([!done] * 1.5)
    while !done {
        reward 1
        done = flip(2/3)
    }
}
\end{lstlisting}
        \caption{HeyVL code to verify that the expected runtime of the modified web server example (\Cref{fig:example-program-b}) is upper-bounded by $2$.}
        \label{fig:example-programs-heyvl-b}
    \end{subfigure}
    \vskip\baselineskip
    \begin{subfigure}[b]{0.45\textwidth}
\begin{lstlisting}[language=HeyVL,linerange={26-38}]
coproc original_moment2() -> (done: Bool)
    pre 6
    post 0
{
    done = false
    var tau: UInt = 0
    @invariant([!done] * (6 + 4 * tau))
    while !done {
        reward 2 * tau + 1
        tau = tau + 1
        done = flip(0.5)
    }
}
\end{lstlisting}
        \caption{HeyVL code to verify that the second moment of the runtime of the unmodified web server example (\Cref{fig:example-program-a}) is upper-bounded by $6$.}
        \label{fig:example-programs-heyvl-c}
    \end{subfigure}
    \hfill
    \begin{subfigure}[b]{0.45\textwidth}
\begin{lstlisting}[language=HeyVL,linerange={40-53}]
coproc caching_moment2() -> (done: Bool)
    pre 4.75
    post 0
{
    done = false
    var tau: UReal = 0.5
    reward 0.25
    @invariant([!done] * (4.5 + 3 * (tau - 0.5)))
    while !done {
        reward 2 * tau + 1
        tau = tau + 1
        done = flip(2/3)
    }
}
\end{lstlisting}
        \caption{HeyVL code to verify that the second moment of the runtime of the modified web server example (\Cref{fig:example-program-b}) is upper-bounded by $4.75$.}
        \label{fig:example-programs-heyvl-d}
    \end{subfigure}
    \caption{HeyVL code to verify claims about \Cref{fig:example-programs}.}
    \label{fig:example-programs-heyvl}
\end{figure}

\begin{figure}[p]
    \centering
\begin{lstlisting}[language=HeyVL]
coproc random_walk(init_x: UInt) -> (x: UInt)
    pre ite(init_x > 1, init_x - (init_x % 2), 0)
    post 0
{
    x = init_x
    @invariant(ite(x > 1, x - (x % 2), 0))
    while x > 1 {
        var prob_choice: Bool = flip(0.75)
        if prob_choice {
            x = x - 2
        } else {
            x = x + 2
        }
        reward 1
    }
}
\end{lstlisting}
    \caption{HeyVL code representing a biased random walk with initial position \texttt{init\_x}.}
    \label{fig:biased-random-walk-original}
\end{figure}

\begin{figure}[p]
    \centering
\begin{lstlisting}[language=HeyVL]
coproc random_walk(init_x: UInt) -> (x: UInt)
    pre init_x*init_x + 3*init_x
    post 0
{
    x = init_x
    var tau: UInt = 0
    @invariant(x*x + 2*x*tau + 3*x)
    while x > 1 {
        var prob_choice: Bool = flip(0.75)
        if prob_choice {
            x = x - 2
        } else {
            x = x + 2
        }
        reward (tau+1)*(tau+1) - tau*tau
        tau = tau + 1
    }
}
\end{lstlisting}
    \caption{HeyVL code after applying transformation function $f(x) = x^2$ to compute the second moment of the time to reach $x \leq 1$.}
    \label{fig:biased-random-walk-second-moment}
\end{figure}

\begin{figure}[p]
    \centering
\begin{lstlisting}[language=HeyVL]
domain Exponentials {
    func pow(base: UReal, exponent: UInt): UReal = ite(exponent == 0, 1, base * pow(base, exponent - 1))
}

coproc original_server_cdf(N: UInt) -> (done: Bool)
    pre pow(5/10, N - 1)
    post 0
{
    var tau: UInt = 0
    reward [0 >= N]

    done = false
    @invariant(
        ite(!done && tau + 1 <= N,
            pow(0.5, N - (tau + 1)),
            0
        )
    )
    while !done {
        reward [tau + 1 == N]
        tau = tau + 1
        done = flip(0.5)
    }
}
\end{lstlisting}
    \caption{HeyVL code for \Cref{ex:cdf-webserver}, verifying upper bounds on the cumulative distribution function of the web server example from \Cref{fig:example-program-a}.}
    \label{fig:cdf-webserver-heyvl}
\end{figure}

\begin{figure}[p]
    \centering
\begin{lstlisting}[language=HeyVL]
domain Exponentials {
    func pow(base: UReal, exponent: UInt): UReal = ite(exponent == 0, 1, base * pow(base, exponent - 1))
}

coproc expected_excess(p: UReal) -> ()
    pre !?(p > 0 && p < 1)
    pre 1/p
    post 0
{
    var done: Bool = false
    @invariant(ite(!done, 1/p, 0))
    while !done {
        done = flip(p)
        reward 1
    }
}
\end{lstlisting}
    \caption{HeyVL code representing the original program from \Cref{ex:loop-splitting-excess}.}
    \label{fig:expected-excess-original}
\end{figure}

\begin{figure}[p]
    \centering
\begin{lstlisting}[language=HeyVL]
domain Exponentials {
    func pow(base: UReal, exponent: UInt): UReal = ite(exponent == 0, 1, base * pow(base, exponent - 1))
}

coproc expected_excess(N: UInt, p: UReal) -> ()
    pre !?(p > 0 && p < 1)
    pre pow(1-p, N)/p
    post 0
{
    var tau: UInt = 0
    var done: Bool = false
    @invariant(ite(!done, pow(1-p, N-tau)/p, 0))
    while !done {
        reward [tau >= N]
        tau = tau + 1
        done = flip(p)
    }
}
\end{lstlisting}
    \caption{HeyVL code after applying transformation function $f(x) = x \monus N$.}
    \label{fig:expected-excess-single}
\end{figure}

\begin{figure}[p]
    \centering
\begin{lstlisting}[language=HeyVL]
domain Exponentials {
    func pow(base: UReal, exponent: UInt): UReal = ite(exponent == 0, 1, base * pow(base, exponent - 1))
}

coproc expected_excess(N: UInt, p: UReal) -> ()
    pre !?(p > 0 && p < 1)
    pre pow(1-p, N)/p
    post 0
{
    var tau: UInt = 0
    var done: Bool = false
    @invariant(ite(!done && tau < N, pow(1-p, N-tau)/p, ite(!done, 1/p, 0)))
    while !done && (tau < N) {
        done = flip(p)
        reward 0
        tau = tau + 1
    }
    @invariant(ite(!done, 1/p, 0))
    while !done {
        done = flip(p)
        reward 1
        tau = tau + 1
    }
}
\end{lstlisting}
    \caption{HeyVL code after splitting the loops and applying transformation function $f(x) = x \monus N$.}
    \label{fig:expected-excess-split}
\end{figure}

\begin{figure}[p]
    \centering
\begin{lstlisting}[language=HeyVL,mathescape]
domain MGFs {
    // uninterpreted function representing $e^{t \cdot x}$
    func mgf(x: UReal): UReal
    axiom mgf_zero mgf(0) == 1
    axiom mgf_increasing forall x: UReal, y: UReal. (x < y) ==> (mgf(x) < mgf(y))
}

coproc randomized_assignment_mgf(p: UReal) -> (x: UInt)
    pre !?(p <= 1)
    pre p * mgf(1) + (1-p)
{
    var tau: UInt = 0
    reward mgf(0)

    var prob_choice: Bool = flip(p)
    if prob_choice {
        x = 1
    } else {
        x = 0
    }

    reward mgf(tau + x) - mgf(tau)
    tau = tau + x
}
\end{lstlisting}
    \caption{HeyVL code after applying transformation function $f(x) = e^{t \cdot x}$ to compute the MGF of a coin flip with probability $p$.}
    \label{fig:mgf-coin-flip-heyvl}
\end{figure}

\begin{figure}[p]
    \centering
\begin{lstlisting}[language=HeyVL]
coproc multiple_rewards(p: UReal, q: UReal) -> ()
    pre !?(p <= 1 && q <= 1)
    pre (4 + (p * (2 + (q * 3))))
{
    var tau_memory: UInt = 2
    var tau_runtime: UInt = 2
    reward tau_memory*tau_runtime

    var prob_choice: Bool = flip(p)
    if prob_choice {

        // extra computation
        reward (tau_memory + 0)*(tau_runtime + 1)-(tau_memory*tau_runtime)
        tau_memory = tau_memory
        tau_runtime = tau_runtime + 1

        prob_choice = flip(q)
        if prob_choice {

            // extra memory
            reward (tau_memory + 1)*(tau_runtime + 0)-(tau_memory*tau_runtime)
            tau_memory = tau_memory + 1
            tau_runtime = tau_runtime

        } else {}
    } else {}
}
\end{lstlisting}
    \caption{HeyVL code for a program that combines two reward structures by multiplication.}
    \label{fig:multiple-rewards-multiplication}
\end{figure}

%% file: appendices/fdr_markov_chain.tex
\definecolor{myorange}{RGB}{230,159,0}
\tikzstyle{every node}=[font=\scriptsize]
\tikzset{dice/.style={draw, minimum size=0.6cm, inner sep=0pt, text width=0.6cm, align=center, fill=myorange!50}}
\tikzset{state/.style={draw, circle, minimum size=0.6cm, inner sep=0pt, text width=0.6cm, align=center}}

\newcommand{\diceFaceOne}{
    \begin{tikzpicture}[scale=0.7]
        \fill (0.5, 0.5) circle (2.2pt);
    \end{tikzpicture}
}

\newcommand{\diceFaceTwo}{
    \begin{tikzpicture}[scale=0.7]
        \fill (0.35, 0.65) circle (2.2pt);
        \fill (0.65, 0.35) circle (2.2pt);
    \end{tikzpicture}
}

\newcommand{\diceFaceThree}{
    \begin{tikzpicture}[scale=0.7]
        \fill (0.35, 0.65) circle (2.2pt);
        \fill (0.5, 0.5) circle (2.2pt);
        \fill (0.65, 0.35) circle (2.2pt);
    \end{tikzpicture}
}

\newcommand{\diceFaceFour}{
    \begin{tikzpicture}[scale=0.7]
        \fill (0.3, 0.3) circle (2.2pt);
        \fill (0.7, 0.3) circle (2.2pt);
        \fill (0.3, 0.7) circle (2.2pt);
        \fill (0.7, 0.7) circle (2.2pt);
    \end{tikzpicture}
}

\newcommand{\diceFaceFive}{
    \begin{tikzpicture}[scale=0.7]
        \fill (0.3, 0.3) circle (2.2pt);
        \fill (0.7, 0.3) circle (2.2pt);
        \fill (0.5, 0.5) circle (2.2pt);
        \fill (0.3, 0.7) circle (2.2pt);
        \fill (0.7, 0.7) circle (2.2pt);
    \end{tikzpicture}
}

\newcommand{\diceFaceSix}{
    \begin{tikzpicture}[scale=0.7]
        \fill (0.3, 0.3) circle (2.2pt);
        \fill (0.7, 0.3) circle (2.2pt);
        \fill (0.3, 0.5) circle (2.2pt);
        \fill (0.7, 0.5) circle (2.2pt);
        \fill (0.3, 0.7) circle (2.2pt);
        \fill (0.7, 0.7) circle (2.2pt);
    \end{tikzpicture}
}

\newcommand{\diceOne}{
    \resizebox{11pt}{9pt}{%
        \begin{tikzpicture}
            \node[dice] (d) {\diceFaceOne};
        \end{tikzpicture}
    }
}
\newcommand{\diceTwo}{
    \kern1pt
    \resizebox{11pt}{9pt}{%
        \begin{tikzpicture}
            \node[dice] (d) {\diceFaceTwo};
        \end{tikzpicture}
    }
}
\newcommand{\diceThree}{
    \kern1pt
    \resizebox{11pt}{9pt}{%
        \begin{tikzpicture}
            \node[dice] (d) {\diceFaceThree};
        \end{tikzpicture}
    }
}
\newcommand{\diceFour}{
    \kern1pt
    \resizebox{11pt}{9pt}{%
        \begin{tikzpicture}
            \node[dice] (d) {\diceFaceFour};
        \end{tikzpicture}
    }
}
\newcommand{\diceFive}{
    \resizebox{11pt}{9pt}{%
        \begin{tikzpicture}
            \node[dice] (d) {\diceFaceFive};
        \end{tikzpicture}
    }
}
\newcommand{\diceSix}{
    \resizebox{11pt}{9pt}{%
        \begin{tikzpicture}
            \node[dice] (d) {\diceFaceSix};
        \end{tikzpicture}
    }
}

\begin{tikzpicture}
    \node[] (helper) {};
    \node[state] (s0) [right=0.5 of helper] {$s_0$};

    \node[state] (s1) [above right=0.95 and 0.01 of s0] {$s_1$};
    \node[state] (s2) [below right=0.95 and 0.01 of s0] {$s_2$};

    \node[state] (s3) [above right=0.15 and 0.7 of s1] {$s_3$};
    \node[state] (s4) [below right=0.55 and 0.7 of s1] {$s_4$};
    \node[state] (s5) [above right=0.35 and 0.7 of s2] {$s_5$};
    \node[state] (s6) [below right=0.15 and 0.7 of s2] {$s_6$};

    \node[dice, right=0.6 of s3] (d1) {\diceFaceOne};
    \node[dice, below right=0.25 and 0.7 of s3] (d2) {\diceFaceTwo};

    \node[dice, right=0.6 of s4] (d3) {\diceFaceThree};
    \node[dice, below right=0.25 and 0.7 of s4] (d4) {\diceFaceFour};

    \node[dice, above right=0.25 and 0.7 of s6] (d5) {\diceFaceFive};
    \node[dice, right=0.6 of s6] (d6) {\diceFaceSix};

    \path[->]
        (helper) edge[] node[above=0.01] {$1$} (s0)

        (s0) edge[] node[left=0.01] {$0.5$} (s1)
        (s0) edge[] node[left=0.01] {$0.5$} (s2)

        (s1) edge[] node[above=0.08,pos=0.6] {$0.5$} (s4)
        (s1) edge[] node[above=0.01,pos=0.3] {$0.5$} (s3)

        (s2) edge[bend right=30] node[above=0.01,pos=0.3] {$0.5$} (s5)
        (s2) edge[] node[below=0.01,pos=0.3]  {$0.5$} (s6)

        (s3) edge[] node[above=0.01,pos=0.3] {$0.5$} (d1)
        (s3) edge[] node[below=0.01,pos=0.3] {$0.5$} (d2)

        (s4) edge[] node[above=0.01,pos=0.3] {$0.5$} (d3)
        (s4) edge[] node[below=0.01,pos=0.3] {$0.5$} (d4)

        (s5) edge[] node[below=0.3,pos=0.6] {$0.5$} (s1)
        (s5) edge[bend right=30] node[above=0.03,pos=0.5] {$0.5$} (s2)

        (s6) edge[] node[above=0.01,pos=0.3] {$0.5$} (d5)
        (s6) edge[] node[below=0.01,pos=0.3] {$0.5$} (d6)
        ;

    \path[->]
        (d1) edge [out=25, in=-25,loop,looseness=5] node[right]{1} (d1)
        (d2) edge [out=25, in=-25,loop,looseness=5] node[right]{1} (d2)
        (d3) edge [out=25, in=-25,loop,looseness=5] node[right]{1} (d3)
        (d4) edge [out=25, in=-25,loop,looseness=5] node[right]{1} (d4)
        (d5) edge [out=25, in=-25,loop,looseness=5] node[right]{1} (d5)
        (d6) edge [out=25, in=-25,loop,looseness=5] node[right]{1} (d6)
        ;

\end{tikzpicture}%

%% file: appendices/more_related_work.tex
\section{Details on Related Work}\label{sec:more-related-work}

\subsection{Relation to the Amortized Expected Runtime Calculus}

An extension of the $\symErt$ calculus called the \emph{amortized expected runtime calculus} was presented in~\cite{DBLP:journals/pacmpl/BatzKKMV23}.
Its $\symAert$ transformer computes the expected change of a potential function $\pi \colon \States \to \PosRealsInf$, i.e. an expectation.
By \cite[Theorem 5.4]{DBLP:journals/pacmpl/BatzKKMV23}: $\symAert_\rt\eval{\stmt}(\expa) = \symWp\eval{\stmt}(\expa + \pi) - \pi$.
Thus, $\symAert$ operates on the lattice $\mathbb{A}$ of \emph{amortized runtimes}, which can be negative expectations, bounded below by the negation of the potential function, i.e. $\mathbb{A} = \Set{\expa \colon \States \to (\Reals \cup \infty) \mid \expa \geq -\pi}$.

However, the analysis of the potential function in the $\symAert$ transformer reasons about the expected value of the potential function \emph{on termination} only.
This is implied by their Theorem 5.4: the potential function $\pi$ is only added to expectation on termination (c.f. \Cref{sec:run-times-and-resource-consumption}).

For further illustration, consider the infinite loop $\stmtWhile{\exprTrue}{\stmtAsgn{\rt}{\rt + 1}}$ that uses the assignment $\stmtAsgn{\rt}{\rt + 1}$ to increase the runtime variable $\rt$ by one in each iteration.
We attempt to analyze the change in potential of a counter variable $\rt$, i.e. choose $\pi = \rt$.
We have $\symAert_\rt\eval{\stmtWhile{\exprTrue}{\stmtAsgn{\rt}{\rt + 1}}}(0) = -\rt$, meaning that an infinite loop does not result in any change in the expected runtime.
To see why the loop's semantics evaluates to $-\pi$, one needs only observe that the $\symAert$ semantics of the $\symWhile$ loop is defined as the least fixed point of the function $\Phi_0(X) = X\substBy{\rt}{\rt + 1} + ((\rt + 1) - \rt) = X\substBy{\rt}{\rt + 1} + 1$.
The least element of the lattice $\mathbb{A}_\pi$ is $-\rt$, but it is already a fixed-point since $\Phi_0(-\rt) = (-(\rt + 1)) + 1 = -\rt$.

Hence, the potential function $\pi$ of the $\symAert$ transformer cannot be used to track run-times without $\symTick$ statements and one could apply our transformation with the trivial potential function $\pi = 0$ instead.